\documentclass[reqno,11pt]{amsart}
\usepackage{amsmath, latexsym, amsfonts, amssymb, amsthm, amscd}
\usepackage{graphics,epsf,psfrag}
\numberwithin{equation}{section}

\newtheorem{theorem}{Theorem}[section]
\newtheorem{lemma}[theorem]{Lemma}
\newtheorem{proposition}[theorem]{Proposition}

\newtheorem{remark}[theorem]{Remark}

\newcommand{\ind}{\mathbf{1}}

\newcommand{\R}{\mathbb{R}}
\newcommand{\Z}{\mathbb{Z}}
\newcommand{\N}{\mathbb{N}}
\renewcommand{\tilde}{\widetilde}
\renewcommand{\hat}{\widehat}

\newcommand{\cT}{{\ensuremath{\mathcal T}} }

\newcommand{\bP}{{\ensuremath{\mathbf P}} }
\newcommand{\bE}{{\ensuremath{\mathbf E}} }

\DeclareMathSymbol{\leqslant}{\mathalpha}{AMSa}{"36} 
\DeclareMathSymbol{\geqslant}{\mathalpha}{AMSa}{"3E} 
\DeclareMathSymbol{\eset}{\mathalpha}{AMSb}{"3F}     
\newcommand{\dd}{\,\text{\rm d}}             

\newcommand{\sumtwo}[2]{\sum_{\substack{#1 \\ #2}}} 

\newcommand{\bbE}{{\ensuremath{\mathbb E}} }

\newcommand{\bbN}{{\ensuremath{\mathbb N}} }

\newcommand{\bbP}{{\ensuremath{\mathbb P}} }

\newcommand{\bbZ}{{\ensuremath{\mathbb Z}} }

\newcommand{\ga}{\alpha}
\newcommand{\gb}{\beta}
\newcommand{\gd}{\delta}
\newcommand{\gep}{\varepsilon}       

\newcommand{\go}{\omega}

\makeatletter
\def\captionfont@{\footnotesize}
\def\captionheadfont@{\scshape}

\long\def\@makecaption#1#2{%
  \vspace{2mm}
  \setbox\@tempboxa\vbox{\color@setgroup
    \advance\hsize-6pc\noindent
    \captionfont@\captionheadfont@#1\@xp\@ifnotempty\@xp
        {\@cdr#2\@nil}{.\captionfont@\upshape\enspace#2}%
    \unskip\kern-6pc\par
    \global\setbox\@ne\lastbox\color@endgroup}%
  \ifhbox\@ne 
    \setbox\@ne\hbox{\unhbox\@ne\unskip\unskip\unpenalty\unkern}%
  \fi
  \ifdim\wd\@tempboxa=\z@ 
    \setbox\@ne\hbox to\columnwidth{\hss\kern-6pc\box\@ne\hss}%
  \else 
    \setbox\@ne\vbox{\unvbox\@tempboxa\parskip\z@skip
        \noindent\unhbox\@ne\advance\hsize-6pc\par}%
\fi
  \ifnum\@tempcnta<64 
    \addvspace\abovecaptionskip
    \moveright 3pc\box\@ne
  \else 
    \moveright 3pc\box\@ne
    \nobreak
    \vskip\belowcaptionskip
  \fi
\relax
}
\makeatother
\def\writefig#1 #2 #3 {\rlap{\kern #1 truecm
\raise #2 truecm \hbox{#3}}}

\newcommand{\tf}{\textsc{f}}

\begin{document}

\title[Pinning models and disorder relevance]{Marginal relevance of disorder \\ 
for  pinning models}

\author{Giambattista Giacomin}
\address{
  Universit{\'e} Paris Diderot (Paris 7) and Laboratoire de Probabilit{\'e}s et Mod\`eles Al\'eatoires (CNRS U.M.R. 7599),
U.F.R.                Math\'ematiques, Case 7012 (Site Chevaleret),
                75205 Paris cedex 13, France
}
\email{giacomin\@@math.jussieu.fr}
\author{Hubert Lacoin}
\address{
  Universit{\'e} Paris Diderot (Paris 7) and Laboratoire de Probabilit{\'e}s et Mod\`eles Al\'eatoires (CNRS U.M.R. 7599),
U.F.R.                Math\'ematiques, Case 7012 (Site Chevaleret),
                75205 Paris cedex 13, France
}
\email{lacoin\@@math.jussieu.fr}
\author{Fabio Lucio Toninelli}
\address{CNRS and ENS Lyon, Laboratoire de Physique, 46 All\'ee d'Italie, 
69364 Lyon, France
}
\email{fabio-lucio.toninelli@ens-lyon.fr}
\date{\today}

\begin{abstract}
   The effect of disorder on pinning and wetting models has attracted
  much attention in theoretical physics ({\sl e.g.}
  \cite{cf:FLNO,cf:DHV}). In particular, it has been predicted on the basis of the {\sl Harris criterion} that
  disorder is {\sl relevant} (annealed and quenched model have
  different critical points and critical exponents) if the return
  probability exponent $\ga $, a positive number that characterizes
  the model, is larger than $1/2$. Weak disorder has been predicted to
  be {\sl irrelevant} ({\sl i.e.} coinciding critical points and exponents)
  if $\ga <1/2$.  Recent mathematical work (in particular
  \cite{cf:Ken,cf:DGLT,cf:GT_cmp,cf:GT_prl}) has put these predictions on firm
  grounds.  In {\sl renormalization group} terms, the case $\ga=1/2$ is
  a {\sl marginal} case and there is no agreement in the literature as
  to whether one should expect disorder relevance \cite{cf:DHV} or
  irrelevance \cite{cf:FLNO} at marginality. The question is
  particularly intriguing also because the case $\ga=1/2$ includes the
  classical models of 
two-dimensional wetting of a rough substrate, of pinning of directed
  polymers on a defect line in dimension $(3+1)$ or $(1+1)$ and of
  pinning of an heteropolymer by a point potential in
  three-dimensional space.  Here we prove disorder relevance both for
  the general $\ga=1/2$ pinning model and for the hierarchical version
  of the model proposed in \cite{cf:DHV}, in the sense that we prove a
  shift of the quenched critical point with respect to the annealed
  one.  In both cases we work with Gaussian disorder and we show that
  the shift is at least of order $\exp(-1/\beta^4)$ for $\beta$ small,
  if $\beta^2$ is the disorder variance. 
\end{abstract}

\maketitle

\section{Introduction}

\subsection{Wetting and pinning on a defect line in $(1+1)$-dimensions}
\label{sec:intro1}
The intense activity aiming at understanding phenomena like wetting in
two dimensions \cite{cf:Abraham} 
and pinning of polymers by a defect line \cite{cf:FLN} has
led several people to focus on a class of simplified models based on
random walks. In order to describe more realistic, spatially
inhomogeneous situations, these models include disordered
interactions.  While a very substantial amount of work has been done,
it is quite remarkable that some crucial issues are not only
mathematically open (which is not surprising given the presence of
disorder), but also controversial in the  physics literature.

Let us start by introducing the most basic, and most studied, model in
the class we consider (it is the case considered 
in \cite{cf:FLNO,cf:DHV}, but also in \cite{cf:BM,cf:GN,cf:GS1,cf:GS2,cf:SC,cf:TangChate},  up to some inessential details, although 
the notations used by the various authors are quite different).
Let $S=\{S_0,S_1, \ldots\}$ be a simple
symmetric random walk on $\bbZ$, {\sl i.e.}, $S_0=0$ and $\{ S_n-S_{n-1}\}_{n \in
  \bbN}$ is an IID sequence (with law $\bP$) of random
variables taking values $\pm 1$ with probability $1/2$.  
It  is  better to take a directed walk viewpoint, that is to consider
the process $\{ (n, S_n)\}_{n=0,1, \ldots}$.
This random
walk is the {\sl free model} and we want to understand
the situation where the walk interacts with a substrate or with a defect
line that provides {\sl disordered} ({\sl e.g.}  random)
rewards/penalties each time  the walk hits it (see Fig.~\ref{fig:SRWwetting}).  The walk
may or may not be allowed to take negative values: we call {\sl
  pinning on a defect line} the first case and {\sl wetting of a
  substrate} the second one. 
  It is by now well understood that these two
cases are equivalent and we briefly discuss the wetting case only in
the caption of Figure~\ref{fig:SRWwetting}: the general model we will
consider covers both wetting and pinning cases.  The interaction is
introduced via the Hamiltonian
\begin{equation}
\label{eq:Haus}
  H_{N,\go}(S):=-\sum_{n=1}^{N}\left(\gb\go_n+h-\log \bbE(\exp(\gb \go_1))\right)\ind_{\{S_n=0\}},
\end{equation}
where $N\in 2\N$ is the system size, $h$ (homogeneous pinning
potential) is a real number, $\go:=\{\go_1,\go_2,\ldots\}$ is a
sequence of IID centered random variables with finite exponential
moments (in this work, we will restrict
to the Gaussian case), $\beta\ge0$ is the disorder strength and $\bbE$
denotes the average with respect to $\go$. It will be soon clear what
is the notational convenience in introducing the non-random term $\log
\bbE(\exp(\gb \go_1))$ (which could be absorbed into $h$ anyway).

\begin{figure}[ht]
\begin{center}
\leavevmode
\epsfxsize =10.4 cm
\psfragscanon
\psfrag{0}{$0$}
\psfrag{N}{$N$}
\psfrag{Sn}{$S_n$}
\psfrag{n}{$n$}
\psfrag{t0}{$\!\!\!\!\!\tau_0(=0)$}
\psfrag{t1}{$\tau_1$}
\psfrag{t2}{$\tau_2$}
\psfrag{t3}{$\tau_3$}
\psfrag{t4}{$\tau_4$}
\psfrag{t5}{$\tau_5(=N/2)$}
\psfrag{om2}{$\tilde\go_2$}
\psfrag{om8}{$\tilde\go_8$}
\psfrag{o12}{$\tilde\go_{12}$}
\psfrag{o14}{$\tilde\go_{14}$}
\psfrag{o16}{$\tilde\go_{16}$}
\psfrag{hlabel}{$\tilde\go_{n}\,:=\, \gb \go_n +h -\log \bbE \exp(\gb \go_1)$}
\psfrag{pinning}{\small trajectory of the pinning model}
\psfrag{wetting}{\small trajectory of the wetting model}
\epsfbox{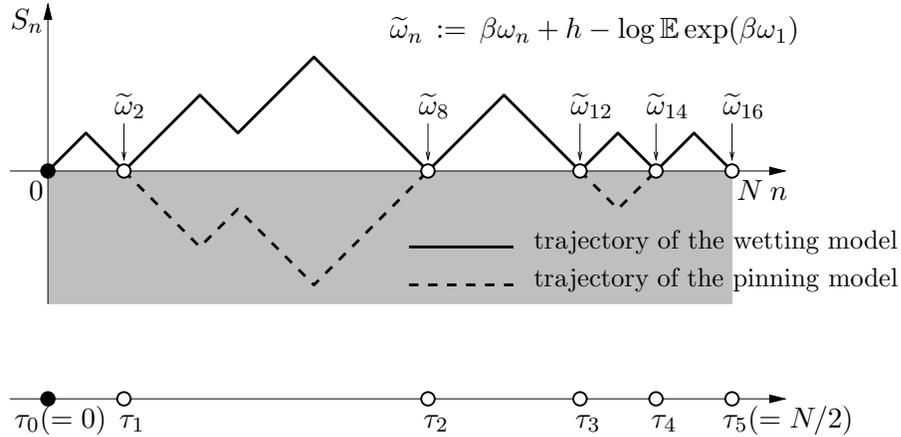}
\end{center}
\caption{In the top a random walk trajectory, pinned at $N$, which
 is not allowed to enter the lower half-plane (the shadowed region
should be regarded as a wall). The trajectory collects the {\sl
charges} $\tilde \go_n$ when it hits the wall.  The question is
whether the rewards/penalties collected pin the walk to the wall or
not. The precise definition of the wetting model is obtained by multiplying
the numerator in the right-hand side of \eqref{eq:Paus} by the indicator
function of the event $\{S_j \ge 0, \, j=1, \ldots, N\}$ (and consequently
modifying the partition function $Z_{N,\go}$).  This model
is actually equivalent to the model  \eqref{eq:Paus} without a wall, 
 whose
trajectories  (dashed line) can visit the lower half plane, provided that
$h $ is replaced by $h-\log 2$ (see \cite[Ch.~1]{cf:Book}). The bottom
part of the figure illustrates the simple but crucial point that the
energy of the model depends only on the location of the points of
contact between walk and wall (or defect line); such points
form a renewal process, giving thus a natural generalized framework in
which to tackle the problem. In order to circumvent the annoying
periodicity two of the simple random walk we set $\tau_0=0$ and
$\tau_{j+1}:= \inf\{ n/2\ge \tau_j:\, S_n=0\}$. From the renewal 
process standpoint, introducing a wall just leads to a {\sl terminating} renewal
(see text).}
\label{fig:SRWwetting}
\end{figure}

The Gibbs measure $\bP_{N,\go}$ for the pinning model is then defined as
\begin{equation}
\label{eq:Paus}
\frac{\dd \bP_{N,\go}}{\dd \bP}(S)=\frac{e^{-H_{N,\go}(S)}\ind_{\{S_N=0\}}}{Z_{N,\go}}
\end{equation}
and of course $Z_{N,\go}:=\bE[\exp(-H_{N,\go}(S))\ind_{\{S_N=0\}}]$,
where $\bE$ denotes expectation with respect to the simple
random walk  measure  $\bP$.
Note that we imposed the boundary condition $S_N=0=S_0$ (just to be
consistent with the rest of the paper).  
It is well known that the model undergoes a
localization/delocalization transition as $h$ varies: if $h$ is
larger than a certain threshold value $h_c(\beta)$ ({\sl quenched
  critical point}) then, under the Gibbs measure, the system is {\sl
  localized}: the contact fraction, defined as
\begin{equation}
\frac1N\bE_{N,\go}\left[\sum_{n=1}^N\ind_{\{S_n=0\}}\right],
\end{equation}
tends to a positive limit for $N\to\infty$. 
On the other hand, for $h<h_c(\gb)$ the system is {\sl delocalized}, {\sl i.e.}, the limit is zero.

The result we just stated is true also in absence of disorder
($\gb=0$) and a remarkable fact for the homogeneous ({\sl i.e.}
non-disordered) model is that it is exactly solvable (\cite{cf:Fisher,cf:Book} and 
references therein).  In particular, 
we know that $h_c(0)=0$, {\sl i.e.}, an arbitrarily small
reward is necessary and sufficient for pinning, and that the free
energy behaves quadratically close to criticality.  If now we consider
the {\sl annealed measure} corresponding to \eqref{eq:Paus}, that is
the model in which one replaces both $\exp(-H_{N,\go}(S))$ and
$Z_{N,\go}$ by their averages with respect to $\go$, one readily
realizes that the annealed model is a homogeneous model, and
precisely the one  we obtain by setting $\gb=0$ in
\eqref{eq:Paus}.  Therefore one finds that the {\sl annealed critical
  point} $h_c^a(\gb)$ equals $0$ for every $\beta$, and that the {\sl annealed free
  energy} $\tf^a(\gb,h)$ behaves, for $h\searrow 0$, like
$\tf^a(\gb,h)\sim const\times h^2$, while it is zero for $h\le 0$.

Very natural questions are: does $h_c(\gb)$ differ
from $h_c^a(\gb)$? Are quenched and annealed critical exponents
different?  As we are going to explain, the first question finds contradictory answers in the
literature, while no clear-cut statement can really be found about the
second. Below we are going to argue that 
these two questions are intimately related, but first we make
 a short detour in order to define a more
general class of models. It is in this more general context that the
role of the disorder and the specificity of the simple random walk
case can be best appreciated.

\subsection{Reduction to renewal-based models}
As argued in the caption of Figure~\ref{fig:SRWwetting}, the basic
underlying process is the {\sl point process} $\tau:=\{\tau_0, \tau_1,
\ldots\}$, which is a renewal process (that is $\{
\tau_{n}-\tau_{n-1}\}_{n \in \N}$ is an IID sequence of 
integer-valued random variables). We set
$K(n):=\bP(\tau_1=n)$. It is well known that, for the simple random
walk case, $\sum_{n\in\N} K(n)=1 $ (the walk is recurrent) and
$K(n)\stackrel{n \to \infty} \sim 1/(\sqrt{4\pi}n^{3/2})$. This
suggests the natural generalized framework of  models based
on discrete renewal processes such that
\begin{equation}
\label{eq:Kintro}
\sum_{n\in \N} K(n) \le 1 \ \text { and } \ K(n)\stackrel{n \to \infty} \sim \frac{C_K}{n^{1+\ga}},
\end{equation}
with $C_K>0$ and $\ga>0$. 
We are of course employing
the standard notation $a_n \sim b_n$ for 
$\lim a_n/b_n =1$.
The case $\sum_{n\in \N} K(n) < 1$ 
 refers to transient (or {\sl terminating}) renewals (of which the wetting 
 case is an example), see also 
Remark \ref{rem:trans} below. This
framework includes for example the simple random walk in $d \ge 3$, for which 
$\sum_{n\in \N} K(n) < 1$ and $\ga =(d/2)-1$, but it is of course
much more general.
We will come back with more details on this model, but let us
just say now that the definition of the Gibbs measure
is given in this case by \eqref{eq:Haus}-\eqref{eq:Paus}, with
$S$ replaced by $\tau$ in the left-hand side and with
the event $\{S_n=0\}$ replaced by the event 
$\{\text{there is } j \text{ such that } \tau_j=n\}$.

\subsection{Harris criterion and disorder relevance: the state of the art}
The questions mentioned at the end of Section \ref{sec:intro1} are
typical questions of disorder relevance, {\sl i.e.}, of stability of
critical properties with respect to (weak) disorder. In
renormalization group language, one is asking whether or not disorder
drives the system towards a new fixed point.  A heuristic tool which
was devised to give an answer to such questions is the  {\sl
Harris criterion} \cite{cf:Harris}, originally proposed for random
ferromagnetic Ising models. The Harris criterion states that disorder
is relevant if the specific heat exponent of the pure system is
positive, and irrelevant if it is negative. In case such critical
exponent is zero (this is called a {\sl marginal case}), the Harris
criterion gives no prediction and a case-by-case delicate analysis is
needed.  Now, it turns out that the random pinning model described
above is a marginal case, and from this point of view it is not
surprising that the question of disorder relevance is not solved yet,
even on heuristic grounds: in particular, the authors of
\cite{cf:FLNO} (and then also \cite{cf:GS1,cf:GS2} and, very recently,
\cite{cf:GN}) claimed that for small $\beta$ the quenched critical
point coincides with the annealed one (with our conventions, this
means that both are zero), while in \cite{cf:DHV} it was concluded
that they differ for every $\beta>0$, and that their difference is of
order $\exp(-const/\beta^2)$ for $\beta$ small (we mention
\cite{cf:BM,cf:SC,cf:TangChate} which  support this 
second possibility). Note that such a quantity is smaller than any
power of $\beta$, and therefore vanishes at all orders in
weak-disorder perturbation theory (this is also typical of marginal
cases).

In an effort to reduce the problem to its core, beyond the
difficulties connected to the random walk or renewal structure, a {\sl
  hierarchical pinning model}, defined on a diamond lattice,
   was introduced in \cite{cf:DHV}. In this
case, the laws of the partition functions for the systems of size $N$
and $2N$ are linked by a simple recursion. The role of $\alpha$ is
played here by a real parameter $B\in(1,2)$, which is related to the
geometry of the hierarchical lattice.  Also in this case, the Harris
criterion predicts that disorder is relevant in a certain regime
(here, $B<B_c:=\sqrt 2$) and irrelevant in another ($B>B_c$), while
$B=B_c$ is the marginal case where the specific heat critical exponent
of the pure model vanishes.  Again, the authors of \cite{cf:DHV}
predicted that disorder is marginally relevant for $B=B_c$, and that
the difference between annealed and quenched critical point behaves
like $\exp(-const/\beta^2)$ for $\beta$ small (they gave also
numerical evidence that the critical exponent is modified by
disorder).

Let us mention that hierarchical models based on diamond lattices have played an important role in 
elucidating the effect of disorder on various statistical mechanics 
models: we mention for instance \cite{cf:DG}.

The mathematical comprehension of the question of disorder relevance
in pinning models has witnessed remarkable progress lately.  First of
all, it was proven in \cite{cf:GT_cmp} that an arbitrarily weak (but
extensive) disorder changes the critical exponent if $\ga>1/2$ (the
analogous result for the hierarchical model was proven in
\cite{cf:LT}). Results concerning the critical points came later: in
\cite{cf:Ken,cf:T_cmp} it was proven that if $\ga<1/2$ then
$h_c(\gb)=0$ (and the quenched critical exponent coincides with the
annealed one) for $\gb$ sufficiently small (the analogous result for
the hierarchical model was given in \cite{cf:GLT}). Finally, the fact
that $h_c(\gb)>0$ for every $\gb>0$ (together with the correct
small-$\gb$ behavior) in the regime where the Harris criterion
predicts disorder relevance was proven in \cite{cf:GLT} in the
hierarchical set-up, and then in \cite{cf:DGLT,cf:AZ} in the
non-hierarchical one. One can therefore safely say that the
comprehension of the relevance question is by now rather solid, {\sl
  except in the marginal case} (of course some problems remain
open, for instance the determination of the value of the quenched
critical exponent in the relevant disorder regime, beyond the bounds
proved in \cite{cf:GT_cmp}).

 \subsection{Marginal relevance of disorder}
 In this work, we solve the question of disorder relevance for the
 marginal case $\alpha=1/2$ (or $B=B_c$ in the hierarchical
 situation), showing that {\sl quenched and annealed critical points
   differ for every disorder strength $\gb>0$}. We also give a
 quantitative bound, $h_c(\gb)\ge \exp(-const/\beta^4)$ for $\gb$
 small, which is however presumably not optimal.  The method we use is
 a non-trivial extension of the {\sl fractional moment -- change of
   measure method} which already allowed to prove disorder relevance
 for $B<B_c$ in \cite{cf:GLT} or for $\ga>1/2$ in \cite{cf:DGLT}. A
 few words about the evolution of this method may be useful to the
 reader. The idea of estimating non-integer moments of the partition
 function of disordered systems is not new: consider for instance
 \cite{cf:BPP} in the context of directed polymers in random
 environment, or \cite{cf:AizM} in the context of Anderson
 localization (in the latter case one deals with non-integer moments
 of the propagator). However, the power of non-integer moments in
 pinning/wetting models was not appreciated until \cite{cf:T_AAP},
 where it was employed to prove, among other facts, that quenched and
 annealed critical points differ for large $\gb$, irrespective of the
 value of $\ga\in(0,\infty)$. The new idea which was needed to treat
 the case of weak disorder (small $\gb$) was instead introduced in
 \cite{cf:GLT,cf:DGLT}, and it is a change-of-measure idea, coupled
 with an {\sl iteration procedure}: one changes the law of the
 disorder $\go$ in such a way that the new and the old laws are very
 close in a certain sense, but under the new one it is easier to prove
 that the fractional moments of the partition function are small. In
 the relevant disorder regime, $\ga>1/2$ or $B<B_c$, it turns out that
 it is possible to choose the new law so that the $\go_n$'s are still
 IID random variables, whose law is simply tilted with respect to the
 original one.  This tilting procedure is bound to fail if applied for
 arbitrarily large volumes, but having such bounds for sufficiently
 large, but finite, system sizes is actually sufficient because of an
 iteration argument (which appears very cleanly in the hierarchical
 set-up).

In order to deal with the marginal case we will instead introduce 
 a long-range anti-correlation structure
 for the $\go$-variables.  Such correlations are carefully chosen in
 order to reflect the structure of the two-point function of the
 annealed model and, in the non-hierarchical case, they are
 restricted, via a coarse-graining procedure 
inspired by \cite{cf:T_cg}, only to suitable {\sl
 disorder pockets}.

We mention also that one of us \cite{cf:Hubert} proved recently that
disorder is marginally relevant in a different version of the
hierarchical pinning model. What simplifies the task in that case is
that the Green function of the model is spatially inhomogeneous and
one can take advantage of that by tilting the $\go$-distributions in a
inhomogeneous way (keeping the $\go$'s independent). The Green
function of the hierarchical model proposed in \cite{cf:DHV} is
instead constant throughout the system and inhomogeneous tilting does
not seem to be of help (as it does not seem to be of help in the
non-hierarchical case, since it does not match with the coarse
graining procedure).

\smallskip 
 The paper is organized as follows: the hierarchical
(resp. non-hierarchical) pinning model is precisely defined in
Section \ref{sec:Hmodel} (resp. in Section \ref{sec:nHmodel}), where
we also state our result concerning marginal relevance of disorder.
Such result is proven in Section \ref{sec:Hproof} in the hierarchical
case, and in Section \ref{sec:nHproof} in the non-hierarchical one.

In order not to hide the novelty of the
idea with technicalities, we restrict ourselves to Gaussian disorder
and, in the case of the non-hierarchical model, we do not treat the
natural generalization where $K(\cdot)$ is of the form
$K(n)={L(n)}/{n^{3/2}}$ with $L(\cdot)$ a slowly varying function
\cite[VIII.8]{cf:Feller2}.  We plan to come back to both issues in a
forthcoming paper \cite{cf:kbodies}.

\section{The hierarchical model}

\label{sec:Hmodel}

Let $1<B<2$.
We study the following iteration which transforms a vector $\{R_n^{(i)}\}_{i\in\N}\in (\R^+)^\N$ into
a new vector $\{R_{n+1}^{(i)}\}_{i\in\N}\in(\R^+)^\N$:
\begin{equation}
  \label{eq:Rn+1}
  R^{(i)}_{n+1}\, =\, \frac{R_n^{(2i-1)} R_n^{(2i)}+(B-1)}B, 
\end{equation}
for $n\in\N\cup\{0\}$ and $i\in\N$.  

In particular, we are interested in
the case in which the initial condition is random and  given by
$R_0^{(i)}=e^{\beta\go_i-\beta^2/2+h}$, with $\go:=\{\go_i\}_{i\in\N}$
a sequence of IID standard Gaussian random variables and $h\in\R, \gb\ge0$.  We
denote by $\bbP$ the law of $\go$ and by $\bbE$
the corresponding average.  In this case, it is immediate to realize
that for every given $n$ the random variables $\{R_n^{(i)}\}_{i\in\N}$ are IID. We
will study the behavior for large $n$ of $X_n:=R_n^{(1)}$.

It is easy to see that the average of $X_n$  
 satisfies the  iteration
\begin{equation}
\label{eq:homomap}
  \bbE({X_{n+1}})\, =\, \frac{(\bbE{X_n})^2+(B-1)}B, 
\end{equation}
with initial condition $\bbE({X_0})=e^h$. The map \eqref{eq:homomap}
has two fixed points: a stable one, $\bbE X_n=(B-1)$, and an unstable
one, $\bbE X_n=1$. This means that if $0\le\bbE X_0<1$ then $\bbE X_n$
tends to $(B-1)$ when $n\to\infty$, while if $\bbE X_0>1$ then $\bbE
X_n$ tends to $+\infty$.

\begin{remark} \rm
  In \cite{cf:DHV} and \cite{cf:GLT}, the model with $B>2$ was
considered. However, the cases $B\in(1,2)$
and $B\in(2,\infty)$ are equivalent. Indeed, if $R_n^{(i)}$ satisfies
\eqref{eq:Rn+1} with $B>2$, it is immediate to see that $\hat
R_n^{(i)}:= R_n^{(i)}/(B-1)$ satisfies the same iteration but with $B$
replaced by $\hat B:=B/(B-1)\in(1,2)$. In this work, we prefer to work
with $B\in(1,2)$ because things turn out to be notationally simpler
({\sl e.g.}, the annealed critical point (defined in the next section) turns
out to be $0$ rather than $\log (B-1)$).  In the following, whenever
we refer to results from \cite{cf:GLT} we give them for $B\in(1,2)$.
\end{remark}

\subsection{Quenched and annealed free energy and critical point}
The random variable $X_n$ is interpreted as the partition function of
the hierarchical random pinning model on a diamond lattice of
generation $n$ (we refer to \cite{cf:DHV} for a clear discussion of
this connection).  The {\sl quenched free energy} is then defined as
\begin{equation}
  \label{eq:F}
  \tf(\beta,h):=\lim_{n\to\infty}\frac1{2^n}\bbE \log X_n.
\end{equation}
In \cite[Th. 1.1]{cf:GLT} it was proven, among other facts, that for every $\beta\ge 0,h\in\R$ the limit
\eqref{eq:F} exists and  it is non-negative. Moreover, $\tf(\beta,\cdot)$ is convex and non-decreasing.
On the other hand, the {\sl annealed free energy} is by definition
\begin{equation}
  \label{eq:Fa}
  \tf^a(\beta,h):=\lim_{n\to\infty}\frac1{2^n}\log \bbE X_n.
\end{equation}
Since the initial condition of \eqref{eq:Rn+1} was normalized so that $\bbE X_0=e^h$, it is
easy to see that the annealed free energy is nothing but the free energy of the non-disordered model:
\begin{equation}
  \label{eq:qz}
\tf^a(\beta,h)=\tf(0,h).  
\end{equation}
Non-negativity of the free energy allows to define the {\sl quenched critical} point in a natural way, as
\begin{equation}
  \label{eq:hc}
  h_c(\beta):=\inf\{h\in\R: \tf(\beta,h)>0\},
\end{equation}
and analogously one defines the {\sl annealed critical point}
$h_c^a(\beta)$. In view of observation \eqref{eq:qz}, one sees that
$h_c^a(\beta)=h_c(0)$.  Monotonicity and convexity of $\tf(\beta,\cdot)$ imply that $\tf(\beta,h)=0$ for
$h\le h_c(\beta)$.

For the annealed system, the critical point
and the critical behavior of the free energy around it are known (see
\cite{cf:DHV} or \cite[Th. 1.2]{cf:GLT}). What one finds is that for
every $B\in(1,2)$ one has $h_c(0)=0$, and there exists $c:=c(B)>0$
such that for all $0\le h\le 1$
\begin{equation}
\label{eq:alpha-1}
  c(B)^{-1}h^{1/\alpha}\le \tf(0,h)\le c(B) h^{1/\alpha},
\end{equation}
where 
\begin{equation}
  \label{eq:alpha}
  \alpha:=\frac{\log(2/B)}{\log 2}\in (0,1).
\end{equation}
Observe that $\alpha$ is decreasing as a function of $B$, and equals $1/2$ for $B=B_c:=\sqrt2$.

\subsection{Disorder relevance or irrelevance}

The main question we are interested in is whether quenched and annealed critical points differ, and 
if yes how does their difference behave for small disorder.
Jensen's inequality, $\bbE\log X_n\le \log \bbE X_n$, implies in particular that $\tf(\beta,h)\le \tf(0,h)$ so that
$h_c(\beta)\ge h_c(0)=0$. Is this inequality strict?

In \cite{cf:GLT} a quite complete picture was given, except in the marginal case $B=B_c$ which was left open:

\begin{theorem}\cite[Th. 1.4]{cf:GLT}
  If $1<B<B_c$, $h_c(\gb)>0$ for every $\gb>0$ and there exists $c_1>0$ such that for $0\le \beta\le 1$
  \begin{equation}
\label{eq:relh}
    c_1 \beta^{2\alpha/(2\alpha-1)}\le h_c(\beta)\le c_1^{-1} \beta^{2\alpha/(2\alpha-1)}.
  \end{equation}

If $B=B_c$ there exists $c_2>0$ such that for $0\le \beta\le 1$
\begin{equation}
\label{eq:margh}
  h_c(\beta)\le \exp(-c_2/\beta^2).
\end{equation}

If $B_c<B<2$ there exists $\beta_0>0$ such that $h_c(\beta)=0$ for every $0<\beta\le \beta_0$.
\end{theorem}

\medskip

The main result of the present work is that in the marginal case, the
two critical points {\sl do differ} for every disorder strength:
\begin{theorem}
\label{th:main}
Let $B=B_c$. For every $0<\beta_0<\infty$ there exists a constant
$0<c_3:=c_3(\beta_0)<\infty$ such that for every $0<\beta\le\beta_0$
  \begin{equation}
\label{eq:main}
    h_c(\beta)\ge \exp(-c_3/\beta^4).
  \end{equation}
\end{theorem}

\section{The non-hierarchical model}\label{sec:nHmodel}

We let $\tau:=\{\tau_0,\tau_1,\ldots\}$ be a renewal process of law
$\bP$, with inter-arrival law $K(\cdot)$, {\sl i.e.}, $\tau_0=0$ and
$\{\tau_i-\tau_{i-1}\}_{i\in\N}$ is a sequence of IID integer-valued
random variables such that
\begin{equation}
  \label{eq:K}
\bP(\tau_1=n)= : K(n)\stackrel{n\to\infty}\sim \frac {C_K}{n^{1+\alpha}},
\end{equation}
with $C_K>0$ and $\alpha>0$. We require that $K(\cdot)$ is a
probability on $\N$, which amounts to assuming that the renewal
process is recurrent. 
 We require also that $K(n)>0$ for every $n\in\N$, but this is inessential 
 and it is just meant to avoid making a certain number of remarks 
 and small detours in the proofs to take care of this point.

As in Section \ref{sec:Hmodel}, $\go:=\{\go_1,\go_2,\ldots\}$ denotes a sequence of
IID standard Gaussian random variables.
For a given system size $N\in\N$, coupling parameters
$h\in \R$, $\beta\ge 0$ and a given disorder realization $\go$ the
partition function of the model is defined by
\begin{equation}
  \label{eq:Znh}
  Z_{N,\go}:=\bE\left[e^{\sum_{n=1}^N(\beta\go_n+h-\beta^2/2)\delta_n}\delta_N\right],
\end{equation}
where $\delta_n:=\ind_{\{n\in\tau\}}$, while the quenched free energy is
\begin{equation}
  \label{eq:F_nh}
  \tf(\beta,h):=\lim_{N\to\infty}\frac1N\bbE\log Z_{N,\go},
\end{equation}
(we use the same notation as for the hierarchical model, since there
is no risk of confusion).  Like for the hierarchical model,
the limit exists and is non-negative \cite[Ch. 4]{cf:Book}, and one
defines the critical point $h_c(\beta)$ for a given $\beta\ge0$
exactly as in \eqref{eq:hc}. Again, one notices that the annealed free
energy, {\sl i.e.}, the limit of $(1/N)\log\bbE Z_{N,\go}$, is nothing but
$\tf(0,h)$, so that the annealed critical point is just $h_c(0)$.

\begin{remark} \rm
  With respect to most of the literature, our definition of the
  model is different (but of course completely equivalent) in that
  usually the partition function is defined as in \eqref{eq:Znh} with
  $h-\beta^2/2$ replaced simply by $h$. 
\end{remark}
The annealed (or pure) model can be exactly solved and in particular
it is well known \cite[Th. 2.1]{cf:Book} that, if $\alpha\ne 1$, there
exists a positive constant $c_K$ (which depends on $K(\cdot)$) such that
\begin{equation}
  \label{eq:annF}
  \tf(0,h)\stackrel{h\searrow0}\sim c_K h^{\max(1,1/\alpha)}.
\end{equation}
In the case $\alpha=1$, \eqref{eq:annF} has to be modified in that the
right-hand side becomes $\phi(1/h)h$ for some slowly-varying function
$\phi(\cdot)$ which vanishes at infinity \cite[Th. 2.1]{cf:Book}. In
particular, note that $h_c(0)=0$ so that $h_c(\beta)\ge0$ by Jensen's
inequality, exactly like for the hierarchical model.

\begin{remark} \rm
\label{rem:trans}
  The assumption of
  recurrence for $\tau$, {\sl i.e.},  $\sum_{n\in\N}K(n)=1$, is
  by no means a restriction. In fact, as it has been observed several
  times in the literature, if $\Sigma_K:=\sum_{n\in\N}K(n)<1$ one can
  define $\tilde K(n):=K(n)/\Sigma_K$, and of course the renewal
  $\tau$ with law $\tilde\bP(\tau_1=n)=\tilde K(n)$ is recurrent. Then,
it is immediate to realize from definition \eqref{eq:F_nh} that
\begin{equation}
\tf(\beta,h)=\tilde \tf(\beta,h+\log \Sigma_K),
\end{equation}
$\tilde \tf$ being the free energy of the model defined as in
\eqref{eq:Znh}-\eqref{eq:F_nh} but with $\bP$ replaced by $\tilde\bP$.
In particular, $h^a_c(\beta)=-\log \Sigma_K$.
This observation allows to apply Theorem \ref{th:main2} below, for
instance, to the case where $\tau$ is the set of returns to the origin
of a symmetric, finite-variance random walk on $\Z^3$ (pinning of a
directed polymer in dimension $(3+1)$): indeed, in this case 
\eqref{eq:K} holds with $\alpha=1/2$. For more details on this issue
we refer to \cite[Ch.~1]{cf:Book}.
\end{remark}

\subsection{Relevance or irrelevance of disorder}
Like for the hierarchical model, the 
question whether $h_c(\beta)$ coincides
or not with $h_c(0)$ for $\beta$ small has been recently solved, {\sl except in the marginal case}
$\ga=1/2$:

\begin{theorem}
If $0<\alpha<1/2$, there exists $\beta_0>0$ such that $h_c(\beta)=0$ for every $0\le \beta\le \beta_0$.
If $\ga=1/2$, there exists a constant $c_4>0$ such that for $\beta\le 1$
\begin{equation}
  h_c(\beta)\le \exp(-c_4/\beta^2).
\end{equation}
If $\ga>1/2$, $h_c(\gb)>0$ for every $\gb>0$ and, if in addition $\ga\ne1$,  there exists a constant $c_5>0$ such that if $\beta\le 1$
\begin{equation}
  c_5\beta^{\max(2\ga/(2\ga-1),2)}\le h_c(\beta)\le  c_5^{-1}\beta^{\max(2\ga/(2\ga-1),2)}.
\end{equation}
If $\ga=1$ there exist a constant $c_6>0$ and a slowly varying function $\psi(\cdot)$ vanishing at infinity such that 
for $\gb\le 1$
\begin{equation}
   c_6\beta^{2}\psi(1/\beta)\le h_c(\beta)\le   c_6^{-1}\beta^{2}\psi(1/\beta).
\end{equation}
\end{theorem}
The results for $\ga\le 1/2$, together with the critical point upper
bounds for $\ga>1/2$, have been proven in \cite{cf:Ken}, and then in
\cite{cf:T_cmp}; the lower bounds on the critical point for $\ga>1/2$
have been proven in \cite{cf:DGLT} 
(the result in \cite{cf:DGLT} is slightly weaker than what we state
here and the case $\ga=1$ was not treated) and then in \cite{cf:AZ}
(with the full result cited here). 

The case $\alpha=0$ has also been
considered, but in that case \eqref{eq:K} has to be replaced by 
$K(n)=L(n)/n$, with
$L(\cdot)$ a function varying slowly at infinity and such that
$\sum_{n\in \N}K(n)=1$. For instance, this corresponds to the case
where $\tau$ is the set of returns to the origin of a symmetric
random walk on $\Z^2$. In this case, it has been shown in \cite{cf:AZ2} that
quenched and annealed critical points coincide for every value of
$\beta\ge0$.

\begin{remark} \rm
Let us recall also that it is proven in \cite{cf:GT_cmp} that, for every $\ga>0$, we have 
\begin{equation}
  \label{eq:smooth}
\tf(\beta,h)\, \le \, \frac{1+\ga}{2\gb^2}\,(h-h_c(\gb))^2, 
\end{equation}
for all $\gb>0,h>h_c(\gb)$: this means that when $\ga>1/2$
disorder is relevant also in the sense that it changes the free-energy critical exponent ({\sl cf.} \eqref{eq:annF}).
The analogous result for the hierarchical model, with $(1+\ga)$ replaced by some constant $c(B)$ in \eqref{eq:smooth},
is proven in \cite{cf:LT}.
\end{remark}

\bigskip

In the present work we prove the following: 
\begin{theorem}
\label{th:main2}
  Assume that \eqref{eq:K} holds with $\alpha=1/2$. 
  For every $\gb_0>0$
  there exists a 
constant $0<c_7:=c_7(\gb_0)<\infty$  such that for $\beta\le \beta_0$
\begin{equation}
  h_c(\beta)\, \ge\,  e^{-c_7/\beta^4}.
\end{equation}
\end{theorem}

\section{Marginal relevance of disorder: the hierarchical case}\label{sec:Hproof}
\subsection{Preliminaries: a Galton-Watson representation
 for $X_n$}
\label{sec:trees}

One can give an expression for $X_n$ which is analogous to that of the
partition function \eqref{eq:Znh} of the non-hierarchical model, and
which is more practical for our purposes. This involves a Galton-Watson tree 
\cite{cf:T-Harris} describing the successive offsprings  of one individual.
The offspring distribution concentrates on $0$ (with probability $(B-1)/B$) and on $2$ (with probability $1/B$). So, at a given 
 generation,
each individual  that is  present has either no descendant or two
descendants, and this independently of any other individual of the
generation. This branching procedure directly maps to a random
tree (see Figure~\ref{fig:treemarg}): the law of such a branching process up
to generation $n$ (the first individual is at generation
$0$) or, analogously, the law of the random tree 
 from the root (level $n$) up to the leaves (level $0$), is denoted by $\bP_n$.
The individuals that are present at the $n^{\textrm{th}}$ generation
are a random subset ${\mathcal R}_n $ of 
$\{ 1, \ldots, 2^n\}$. We set $\gd_j:= \ind _{j \in {\mathcal R}_n}$.
Note that the mean offspring size is $2/B>1$, so that the Galton-Watson process
is supercritical. 

The following procedure on the standard binary graph $\cT^{(n)}$ of
depth $n+1$ (again, the root is at level $n$ and the leaves, numbered
from $1$ to $2^n$, at level $0$) is going to be of help too.  Given
$\mathcal I\subset \{1,\ldots,2^n\}$, let $\mathcal T^{(n)}_{\mathcal
  I}$ be the subtree obtained from $\mathcal T^{(n)}$ by deleting all
edges except those which lead from leaves $j\in\mathcal I$ to the
root. Note that, with the offspring distribution we consider, in general
$\mathcal T^{(n)}_{\mathcal I}$ is not a realization of 
the $n$-generation Galton-Watson tree (some individuals may have
just one descendant in $\mathcal T^{(n)}_{\mathcal I}$, 
see Figure~\ref{fig:treemarg}).

 Let $v(n,\mathcal I)$ be the number of nodes
in $\mathcal T^{(n)}_{\mathcal I}$, with the convention that leaves are not counted as nodes, while the root is.

\begin{figure}[ht]
\begin{center}
\leavevmode
\epsfxsize =10.7 cm
\psfragscanon
\psfrag{i}{$4$}
\psfrag{j}{$6$}
\psfrag{k}{$13$}
\psfrag{0}[r]{\small \sl level $0$}
\psfrag{L}[r]{\small \sl (the leaves)}
\psfrag{R}[r]{\small \sl (the root)}
\psfrag{1}[r]{\small  \sl level $1$}
\psfrag{2}[r]{\small  \sl level $2$}
\psfrag{r}[r]{\small  \sl level $4$}
\psfrag{3}[r]{$\ldots\ \ $ }
\epsfbox{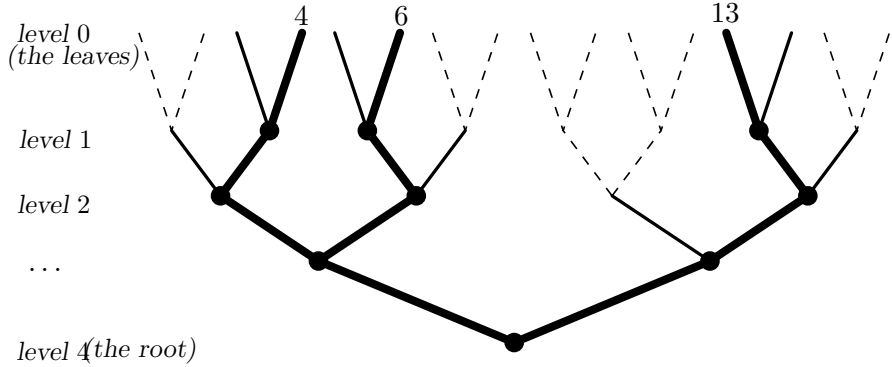}
\end{center}
\caption{
The thick solid lines in the figure form the
tree $\mathcal T^{(4)}_{\{4,6,13\}}$, which is a subtree of the
 binary tree $\mathcal T^{(n)}$ ($n=4$). Note that
$\mathcal T^{(4)}_{\{4,6,13\}}$ is {\sl not} a possible 
realization of the Galton-Watson tree, while it becomes so if we complete it
by adding the thin solid lines.
At level $0$ there
are the leaves; the nodes of $\mathcal T^{(4)}_{\{4,6,13\}}$
are marked by dots. $\mathcal T^{(4)}_{\{4,6,13\}}$ contains
$v(4,\{4,6,13\})=9$ nodes. In terms of Galton-Watson offsprings,
for the ({\sl completed}) trajectory above
$\mathcal R _4=\{3,4,5,6,13,14\}$. Moreover,
computing the average in  
\eqref{eq:ffv1} means computing the probability that the 
realization of 
the Galton-Watson 
tree  contains $\mathcal T^{(n)}_{\mathcal I}$
as a subset: but this simply means  requiring that 
the individuals at the nodes of $\mathcal T^{(n)}_{\mathcal I}$
have two children and the expression \eqref{eq:ffv1} becomes clear.}
\label{fig:treemarg}
\end{figure}

\medskip

\begin{proposition}
\label{th:deltas}
  For every $n\ge0$ we have
\begin{equation}
  \label{eq:Rn}
  X_n\, = \,\bE_n\left[
  e^{\sum_{i=1}^{2^n}(\beta \go_i+h-\beta^2/2)\delta_i}\right].
\end{equation}
For every $n\ge 0$ and $\mathcal I\subset \{1,\ldots,2^n\}$, one has
\begin{equation}
\label{eq:ffv1}
  \bE_n\bigg[\prod_{i\in\mathcal I}\delta_i\bigg]\, =\, B^{-v(n,\mathcal I)}.
\end{equation}
In particular, $\bE_n[\delta_i]=B^{-n}$ for every $i=1,\ldots,2^n$,
{\sl i.e.}, the Green function is constant throughout the system. 
\end{proposition}
\medskip



\begin{proof}[Proof of Proposition \ref{th:deltas}]
The right-hand side in  \eqref{eq:Rn} 
for $n=0$ is equal to $\exp(h-\gb^2/2+ \gb \go_1)$.
Moreover, at the $(n+1)^{\textrm{th}}$ generation the branching process
either contains only the initial individual (with probability $(B-1)/B$)
or the initial individual  has two children, which we may look at as 
initial  individuals
of two independent 
Galton-Watson trees containing $n$ new generations.
We therefore have  that the basic recursion  \eqref{eq:Rn+1}
is satisfied. 

The second fact, \eqref{eq:ffv1}, is a direct consequence of the definitions
(see also the caption of Figure~\ref{fig:treemarg}).
\end{proof}

\begin{remark} \rm
The representation we have introduced in this section shows in particular
that $\bbE X_n$ is just the generating function of $\vert \mathcal R_n\vert$
and the free energy $\tf(0, h)$ is therefore a natural quantity
for the 
Galton-Watson process: and in fact  $1/\ga$ ($\ga$ given in 
\eqref{eq:alpha}) appears in the original works on branching processes by
 T.~E.~Harris
(of course not to be confused with  A.~B.~Harris, who proposed the 
disorder relevance criterion on which we are focusing in this work).
\end{remark}

\subsection{The proof of Theorem~\ref{th:main}}
While the discussion of the previous section is valid for every
$B\in(1,2)$,  
now we have to assume $B=B_c=\sqrt 2$. However some of the steps are still
valid in general and we are going to replace $B$ with $B_c$ only 
when it is really needed.  
 The
proof is split into three subsections: the first  introduces
the fractional moment method and reduces the statement we want to
prove, which is a statement on the limit $n \to \infty$ behavior of
$X_n$, to finite-$n$ estimates. The estimates are provided in the
second and third subsection.

\subsubsection{The fractional moment method}
Let $U_n^{(i)}$ denote the quantity $[R_n^{(i)}-(B-1)]_+$ where
$[x]_+=\max(x,0)$. 
Using the inequality
\begin{equation}
  [rs+r+s]_+\le [r]_+[s]_++[r]_++[s]_+,
\end{equation}
which holds whenever $r,s\ge-1$,
it is easy to check that \eqref{eq:Rn+1} implies

\begin{equation}
\label{eq:recU}
 U^{(i)}_{n+1}\le \frac{U_n^{(2i-1)}U_n^{(2i)}+(U_n^{(2i-1)}+U_n^{(2i)})(B-1)}{B}.
\end{equation}
Given $0<\gamma<1$, we define $A_n:=\bbE([X_n-(B-1)]_+^\gamma)$. From \eqref{eq:recU} above and by using the
fractional inequality 
\begin{equation}
  \label{eq:fracineq}
  \left(\sum a_i\right)^{\gamma}\, \le\,  \sum a_i^{\gamma},
\end{equation}
which holds whenever $a_i\ge0$, we derive
\begin{equation}
 A_{n+1}\le \frac{A_n^2+2(B-1)^{\gamma}A_n}{B^{\gamma}}.
\end{equation}
One readily sees now 
 that,  if there exists some integer $k$ such that 
\begin{equation}
\label{eq:k}
A_k\, < \, B^{\gamma}-2(B-1)^{\gamma},
\end{equation}
then $A_n$ tends to zero as $n$ tends to infinity (this statement is
easily obtain by studying the fixed points of the function $x\mapsto
{(x^2+2(B-1)^{\gamma}x)}/{B^{\gamma}}$). 
On the other hand,
\begin{equation}
\label{eq:LB0411}
\bbE\left[ X_n ^\gamma \right] \, \le \, \bbE\left([X_n-(B-1)]_++(B-1)\right)^\gamma\le 
(B-1)^\gamma + A_n,
\end{equation}
and therefore \eqref{eq:k}  implies that $\tf (\gb, h)=0$
since, by Jensen inequality, we have
\begin{equation}
\label{eq:JensAn}
\frac 1{2^n} \bbE \log X_n  \, \le \, \frac 1{2^n\gamma} \log  \bbE\left[ X_n ^\gamma \right].
\end{equation}
Note that, to establish $\tf (\gb, h)=0$, 
it suffices to prove that 
  $\limsup_n 2^{-n}\log A_n \le 0$,
  hence our approach yields a substantially stronger piece of
  information, {\sl i.e.} that the fractional moment $A_n$ does go to
  zero.

\medskip 

In order to find a $k$ such that \eqref{eq:k} holds we introduce a new
probability measure $\tilde \bbP$ (which is going to depend on $k$) such
that $\tilde \bbP$ and $\bbP$ are equivalent, that is mutually
absolutely continuous.  By H\"older's inequality 
applied for $p=1/\gamma$ and $q=1/(1-\gamma)$ we have
\begin{equation}
\begin{split}
\label{eq:hold}
  A_k\, =\, \tilde \bbE\left[ \frac{\dd \bbP}{\dd \tilde
      \bbP}\;[X_k-(B-1)]_+^{\gamma}\right] \le
  \left(\bbE\left[\left(\frac{\dd \bbP}{\dd \tilde
          \bbP}\right)^{\frac{\gamma}{1-\gamma}}\right]\right)^{1-\gamma}
  \left(\tilde\bbE\left[ [X_k-(B-1)]_+\right]\right)^{\gamma},
\end{split}
\end{equation}
and a sufficient condition for \eqref{eq:k} is  therefore that
\begin{equation}
\label{eq:cond1}
 \tilde\bbE\left[ [X_k-(B-1)]_+\right]\, \le\,  \left(\bbE\left[\left(\frac{\dd \bbP}{\dd \tilde \bbP}\right)^{\frac{\gamma}{1-\gamma}}\right]\right)^{1-\frac{1}{\gamma}}\left(B^{\gamma}-2(B-1)^{\gamma}\right)^{\frac{1}{\gamma}}.
\end{equation}

Let $x^{(0)}_n$ be obtained applying $n$ times the annealed iteration
$x\mapsto (x^2+(B-1))/B$ to the initial condition $x^{(0)}_0=0$. One
has that $x_n^{(0)}$ approaches monotonically the stable fixed point
$(B-1)$. Since the coefficients in the iteration \eqref{eq:Rn+1} are
positive, one has for every $h,\gb,\go$ that $X_n
\ge x_n^{(0)} \stackrel{n\to\infty}\nearrow B-1$ (this is a deterministic bound) and therefore, for any given $\zeta>0$,
one can find an integer $n_\zeta$ such that if $n\ge n_\zeta$ we have
\begin{equation}
\tilde\bbE\left[ [X_n-(B-1)]_+\right]\, \le \, \tilde\bbE\left[ X_n-(B-1)\right]+\frac \zeta 4.
\end{equation}
Moreover, since
$\left(B^{\gamma}-2(B-1)^{\gamma}\right)^{\frac{1}{\gamma}}- (2-B) \stackrel{\gamma
\nearrow 1}\sim 
-c_B (1-\gamma)$ for some $c_B>0$, one can find  
$\gamma=\gamma_\zeta$ such that $\left(B^{\gamma}-2(B-1)^{\gamma}\right)^{\frac{1}{\gamma}}\ge 2-B-\zeta/4$.
At this point,
if $\gamma=\gamma_\zeta$, $k\ge n_\zeta$ and if $\tilde\bbP$ is such that
\begin{equation}
\left(\bbE\left[\left(\frac{\dd \bbP}{\dd \tilde \bbP}\right)^{\frac{\gamma}{1-\gamma}}\right]\right)^{1-\frac{1}{\gamma}}\ge 1-\frac {\zeta}4.
 \end{equation}
 (recall that $\tilde \bbP$ depends on $k$)
and $\tilde\bbE [X_k]\le 1-\zeta$ then \eqref{eq:cond1} is satisfied
and $\tf(\gb, h)=0$.

We sum up what we have obtained:
\medskip

\begin{lemma}\label{th:moment}Let $\zeta>0$ and
choose  $\gamma(=\gamma_\zeta)$ and   $n_\zeta$
as above. If there exists $k\ge n_\zeta$ and 
a probability measure $\tilde \bbP$ (such that $\bbP$
and $\tilde \bbP$ are equivalent probabilities) such that
\begin{equation}\label{eq:close}
\left(\bbE\left[\left(\frac{\dd \bbP}{\dd \tilde \bbP}\right)^{\frac{\gamma}{1-\gamma}}\right]\right)^{1-\frac{1}{\gamma}}\, \ge \, 1-\frac{\zeta}4 , 
 \end{equation}
and 
\begin{equation}
 \tilde \bbE [X_k]\le 1-\zeta,
\end{equation}
then the free energy is equal to zero.
\end{lemma}

\subsubsection{The change of measure}

In order to use wisely the result of the previous section, we have to
find a measure $\tilde \bbE:=\tilde \bbE_n$ on the environment which
is, in a sense, close to $\bbE$ ({\sl cf.} \eqref{eq:close}), and that
lowers significantly the expectation of $X_n$. In \cite{cf:GLT} we
introduced the idea of changing the mean of the $\go$-variables,
while keeping their IID character. This strategy was enough to prove
disorder relevance for $B<B_c$, but it is not effective in the
marginal case $B=B_c$ we are considering here. Here, instead, we
choose to introduce {\sl weak, long range} negative correlations
between the different $\go_i$ without changing the laws of the
1-dimensional marginals. As it will be clear, the covariance structure we
choose reflects the hierarchical structure of the model we are
considering.

 In the
sequel we take $h\ge h_c(0)=0$.

We define $\tilde \bbP_n$ by stipulating that the variables $\go_i,i>2^n,$ are still 
IID standard Gaussian independent of $\go_1, \ldots, \go_{2^n}$, while $\go_1, \ldots, \go_{2^n}$
are Gaussian, centered,  and with covariance matrix
\begin{equation}
\label{eq:covar}
  C:=I-\gep V,
\end{equation}
where $I$ is the $2^n\times 2^n$ identity matrix, $\gep>0$ and $V$ is
a  symmetric $2^n\times 2^n$ matrix with zero diagonal terms 
and with positive off-diagonal terms
($\gep$ and $V$
will be specified in a moment). 

The choice $V_{ii}=0$ implies of course $\text{Trace}(V)=0$, and we
are also going to impose that the Hilbert-Schmidt norm of $V$ verifies 
$\Vert V \Vert ^2 := \sum_{i,j } V_{i,j}^2 =\text{Trace}(V^2)=1$. This in particular
implies that $C$ is positive definite (so that $\tilde \bbP_n$
exists!) as soon as $\gep<1$: this is because $\Vert V \Vert$, being a
matrix norm, dominates the spectral radius of $V$.

Now, still without choosing $V$ explicitly,  we compute a lower bound for the left--hand side of \eqref{eq:close}. 
The mutual density of $\tilde \bbP_n$ and $\bbP$  is
\begin{equation}
  \label{eq:densit}
  \frac{\dd \tilde \bbP_n}{\dd \bbP}(\go)\, =
  \, 
  \frac{e^{-1/2 ((C^{-1}-I)\go,\go)}}{\sqrt{\det C}},
\end{equation}
with the notation $(Av,v):=\sum_{1\le i,j\le 2^n}A_{ij}v_i v_j$,
and therefore a straightforward Gaussian computation gives
\begin{equation}
  \label{eq:dets}
  \left( \bbE \left[ \left(\frac{\dd\bbP}{\dd\tilde \bbP_n}\right)^{\gamma/(1-\gamma)}\right]\right)^{1-1/\gamma}\, =\, 
  \frac{(\det[I-(\gep/(1-\gamma)) V])^{(1-\gamma)/(2\gamma)}}{(\det C)^{1/(2\gamma)}}.
\end{equation}
If we want to prove a lower bound of the type \eqref{eq:close}, a
necessary condition is of course that the numerator in \eqref{eq:dets}
is  positive: this is ensured by requiring $\gep < 1-\gamma$. For the
next computation we are going to require also that $ \gep/(1-\gamma)
\le 1/2$: we are going in fact to use that $\log (1+x) \ge x -x^2$ if
$x \ge -1/2$, and $\text{Trace}(V)=0$ to obtain that
\begin{multline}
\label{eq:dano}
\det\left[I-(\gep/(1-\gamma)) V\right] \, =\,
\exp \left( \text{Trace}( \log (I-(\gep/(1-\gamma))V)) \right)\\
 \ge \, 
\exp\left( -\frac{\gep^2}{(1-\gamma)^2} \Vert V \Vert ^2 \right),
\end{multline}
while
$\log (1+x) \le x$ and the traceless character of $V$ directly imply $\det C \le 1$ so that finally
\begin{equation}
  \label{eq:dets-est}
  \left(  \bbE\left[\left(\frac{\dd\bbP}{\dd\tilde \bbP_n}\right)^{\gamma/(1-\gamma)}\right]\right)^{1-1/\gamma}\, \ge\, 
  \exp \left( -\frac{\gep^2}{2\gamma(1-\gamma)} \right).
\end{equation}

\smallskip

Next, we estimate the expected value of $X_n$ under the modified measure: from \eqref{eq:Rn} we see that
\begin{equation}\begin{split}
\label{eq:EZN}
\tilde \bbE_n X_n&\, =\, \bE_n\left[e^{(h-(\beta^2/2))\sum_{i=1}^{2^n}
  \delta_i}\,\tilde\bbE_n e^{\sum_{i=1}^{2^n}
  \gb\go_i\delta_i}\right]\\
&\, =\, \bE_n\left[e^{-\gep(\beta^2/2) (V\delta,\delta)+\sum_{i=1}^{2^n}
  h\delta_i}\right]
  \, \le\, e^{2^n h}\,\bE_n\left[e^{-\gep(\beta^2/2)
  (V\delta,\delta)}\right].
\end{split}
\end{equation}

Finally we choose $V$. From \eqref{eq:EZN}, it is not hard to guess that the most convenient choice, subject to the 
constraint $\Vert V\Vert^2=1$, is
\begin{equation}
  \label{eq:V}
  V_{ij}\, =\, \bE_n[\delta_i\delta_j]\bigg /\sqrt{\sum_{1\le i\ne j\le 2^n}(\bE_n[\delta_i\delta_j])^2},
\end{equation}
for $i\ne j$, while we recall that $V_{ii}=0$. 
The normalization in \eqref{eq:V} can be computed with the help of Proposition 
\ref{th:deltas}:
\begin{equation}
\label{eq:2n}
  \sum_{1\le i\ne j\le 2^n}\left(\bE_n[\delta_i\delta_j]\right)^2=
2^n\sum_{1<j\le 2^n}\left(\bE_n[\delta_1\delta_j]\right)^2=
2^n\sum_{1\le a\le n}\frac{2^{a-1}}{B_c^{2(n+(a-1))}}=n.
\end{equation}
In the second equality, we used the fact that there are $2^{a-1}$
values of $1<j\le 2^n$ such that the two branches of the tree
$\mathcal T^{(n)}_{\{1,j\}}$ join at level $a$ ({\sl cf.} the notations of
Section \ref{sec:trees}), and such tree contains $n+a-1$ nodes.

As a side remark, note that if $B_c<B<2$ (irrelevant disorder regime)
the left-hand side of \eqref{eq:2n} instead goes to zero with $n$,
while for $1<B<B_c$ (relevant disorder regime) it diverges
exponentially with $n$.

So, in the end,  our choice for $V$ is: 
\begin{equation}
  \label{eq:VV}
  V_{ij}=\left\{
    \begin{array}{lll}
      {\bE_n[\delta_i\delta_j]}/{\sqrt{n}} &\mbox{if}& i\ne j\\
0&\mbox{if}& i=j.
    \end{array}
\right.
\end{equation}

\subsubsection{Checking the conditions of Lemma \ref{th:moment}}

To conclude the proof of Theorem \ref{th:main} we have to show that if
$\beta\le \beta_0$ and $h\le \exp(-c_3/\beta^4)$ (and provided that
$c_3= c_3(\beta_0)$ is chosen large enough) the conditions of Lemma
\ref{th:moment} are satisfied. The main point is therefore to
estimate the expectation of $X_n$ under $\tilde\bbP_n$.

Recalling that ({\sl cf.} \eqref{eq:EZN})
\begin{equation}
\label{eq:q2}
  \tilde \bbE_n X_n\le \bE_n \left[e^{-(\beta^2/2)\gep\sum_{1\le i\ne j\le 2^n}\delta_i\delta_j
    \frac{\bE_n[\delta_i\delta_j]}{\sqrt{n}}}\right]e^{2^n h}
,
\end{equation}
we define
\begin{equation}
Y_n:=\sum_{1\le i\ne j\le 2^n}\delta_i\delta_j
    \frac{\bE_n[\delta_i\delta_j]}{n}.
\end{equation}
Thanks
to \eqref{eq:2n}, we know that $\bE_n (Y_n)=1$,
so that the Paley-Zygmund inequality gives
\begin{equation}
  \label{eq:PZ}
  \bP_n \left(Y_n\ge 1/2\right) \,=\, 
  \bP_n \left(Y_n\ge (1/2)\bE_n(Y_n)\right)\, \ge 
  \, \frac{(\bE_n(Y_n))^2}{4\,\bE_n(Y_n^2)}=
  \frac{1}{4\,\bE_n(Y_n^2)}.
\end{equation}
We need therefore the following estimate, which will be proved at the end of the section:
\begin{lemma} \label{th:secmom}
We have: 
\begin{equation}
(1\le )\;\mathcal K:=\sup_{n}  \bE_n[Y_n^2]<\infty.
\end{equation}
\end{lemma}

Together with \eqref{eq:PZ} this implies
\begin{equation}
 \bP_n[Y_n\ge 1/2]\,\ge  \, \frac{1}{4\mathcal K}, 
\end{equation}
so that, for all $n\ge 0$,
\begin{equation}
\label{eq:bd}
\begin{split}
\bE_n \left[e^{-(\beta^2/2)\gep\sum_{1\le i\ne j\le 2^n}\delta_i\delta_j
  \frac{\bE_n[\delta_i\delta_j]}{\sqrt{n}}}\right]\, &=\, 
\bE_n\left[e^{-\frac{\sqrt{n}\gb^2\gep}2 Y_n}\right]\\
&\le \, 1-\frac{1}{4\mathcal K}\left(1-4
  \mathcal K\exp\left(-\frac{\sqrt{n}\gb^2\gep}{4}\right)\right).
  \end{split}
\end{equation}
We fix $\zeta:=1/(40\mathcal K)$
and we choose $\gamma=\gamma_\zeta$ ({\sl cf.} Lemma \ref{th:moment}) and 
$\gep$ in \eqref{eq:covar} small enough so that ({\sl cf.} \eqref{eq:dets-est})
\begin{equation}\label{eq:upb2}
  \left[  \bbE\left(\frac{\dd\bbP}{\dd\tilde \bbP_n}\right)^{\gamma/(1-\gamma)}\right]^{1-1/\gamma}\, \ge\, 
\exp \left( -\frac{\gep^2}{2\gamma(1-\gamma)} \right)\ge 1-\frac{\zeta}{4}.
\end{equation}
Then one can check with the help of \eqref{eq:bd} that for $n\ge {50\mathcal K}/{(\gb^4\gep^2)}$,
\begin{equation}\label{eq:upb3}
\bE_n \left[e^{-(\beta^2/2)\gep\sum_{1\le i\ne j\le 2^n}\delta_i\delta_j
    \frac{\bE_n[\delta_i\delta_j]}{\sqrt{n}}}\right]\, \le \,  1-3\zeta.
\end{equation}
We choose $n=n_\gb$ in
$\left[\frac{50\mathcal K}{\gb^4\gep^2},\frac{50\mathcal K}{\gb^4\gep^2}+1\right)$ and
$h=\zeta 2^{-n}$. If $\gep$ has been chosen  small enough above (how small,
 depending only 
on $\beta_0$), this guarantees that $n\ge n_\zeta$, where $n_\zeta$
was defined just before Lemma \ref{th:moment}. Injecting
\eqref{eq:upb3} in \eqref{eq:q2} finally gives
\begin{equation}
 \tilde \bbE[X_n]\le (1-3\zeta)e^{\zeta}\le 1-\zeta.
\end{equation}

The two conditions of Lemma \ref{th:moment} are therefore verified,
which ensures that the free energy is zero for this value of $h$. In
conclusion, for every $\gb\le \beta_0$ we have proven that
\begin{equation}
 h_c(\gb)\,\ge
 \,  \zeta\, 2^{-n_{\gb}}\, \ge\,  \frac 1{80\mathcal K }\exp\left(-\frac{50  \mathcal K \log 2}{\gb^4\gep^2} \right),
\end{equation}
for some $\gep=\gep(\beta_0)$ sufficiently small but independent of $\gb$.
\qed

\begin{proof}[Proof of Lemma \ref{th:secmom}]
We have
\begin{equation}
  \label{eq:Y2}
  \bE_n(Y_n^2)=\frac 1{n^2}\sum_{1\le i\ne j\le 2^n}\sum_{1\le k\ne l\le 2^n}\bE_n[\delta_i\delta_j]\bE_n[\delta_k\delta_l]
  \bE_n[\delta_i\delta_j\delta_k\delta_l].
\end{equation}
We will consider only the contribution coming from the terms such that $i\ne k,l$ and $j\ne k,l$. The 
remaining terms can be treated similarly and their global contribution is easily seen to be exponentially small 
in $n$. (For instance,  when $i=k$ and $j=l$ one gets
\begin{equation}
  \frac 1{n^2}\sum_{1\le i\ne j\le 2^n}\bE_n[\delta_i\delta_j]^3
  \, \le\,
   \frac1n\bE_n(Y_n)\max_{1\le i<j\le 2^n}\bE_n[\delta_i\delta_j],
\end{equation}
which is exponentially small in $n$, in view of Theorem \ref{th:deltas}.)

\begin{figure}[ht]
\begin{center}
\leavevmode
\epsfxsize =12 cm
\psfragscanon
\psfrag{a}{(a)}
\psfrag{b}{(b)}
\psfrag{k}{$k$}
\psfrag{0}[r]{level $0$: the leaves}
\psfrag{v}[r]{ $v$}
\psfrag{2}[r]{level $2$}
\psfrag{4}[c]{level $n=4$: the root}
\psfrag{3}[c]{$\ldots$}
\psfrag{r}[c]{level $N$: the root}
\epsfbox{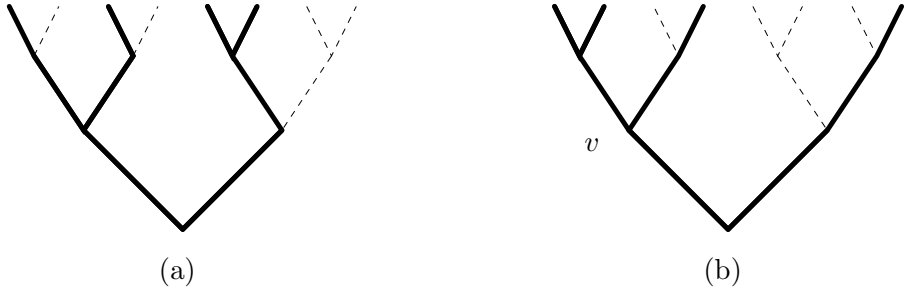}
\end{center}
\caption{The two different possible topologies of the tree $\mathcal
  T^{(n)}_{\{i,j,k,l\}}$.  Case (b) is understood to include also the trees
  where the branch which does not bifurcate is the one on the left, or
  where the sub-branch which bifurcates is the right descendent of the
  node $v$.  We consider only trees where the four leaves are
  distinct, since the remaining ones give a contribution to $\bE_n
  (Y_n^2)$ which vanishes for $n\to\infty$.}
\label{fig:2t}
\end{figure}
From now on, therefore, we assume that $i,j,k,l$ are all distinct. Two cases can occur:
\begin{enumerate}
\item the tree $\mathcal T^{(n)}_{\{i,j,k,l\}}$   (it is better to view it here
has the backbone tree, not as the Galton-Watson tree, see Figure~\ref{fig:treemarg})
has two branches, which
  themselves bifurcate into two sub-branches, {\sl cf.} Fig. \ref{fig:2t}(a)
  for an example.  We call $c$ the level at which the first
  bifurcation occurs ($c=n$ in the example of Fig. \ref{fig:2t}(a)), and
  $a,b$ the levels at which the two branches bifurcate. One has
  clearly $1\le a<c\le n$ and $1\le b<c\le n$. All trees of this form
  can be obtained as follows: first choose a leaf $f_1$,
  between $1$ and $ 2^n$. Then choose $f_2$ among the $2^{a-1}$
  possible ones which join with $f_1$ at level $a$, $f_3$ among the
  $2^{c-1}$ which join with $f_1$ at level $c$ and finally $f_4$ among
  the $2^{b-1}$ which join with $f_3$ at level $b$. 
Clearly we are over-counting the trees (note for example that already 
in the choice of $f_1$ and $f_2$ we are over-counting by a factor $2$), but 
we are only after an {\sl upper bound} for $\bE_n(Y_n^2)$ (the
same remark applies to case (2) below).
We still have
  to specify how to identify $(f_1,f_2,f_3,f_4)$ with a permutation
  of $(i,j,k,l)$. When $(f_1,f_2,f_3,f_4)=(i,j,k,l)$ we get the
  following contribution to \eqref{eq:Y2}:
\begin{equation}
\label{eq:caso11}
\frac1{n^2}  \sum_{1\le a<c\le n}\sum_{1\le b<c\le n}\frac{2^{n+a+b+c-3}}{B_c^{n+a+b+c-3}B_c^{n+a-1} B_c^{n+b-1}},
\end{equation}
where we used Theorem \ref{th:deltas} to write, {\sl e.g.},
$\bE_n[\delta_i\delta_j]=B_c^{-n-a+1}$. Since $B_c=\sqrt 2$ we can
rewrite \eqref{eq:caso11} as
\begin{equation}
  \frac1{\sqrt 2n^2}\sum_{1<c\le n}(c-1)^2 2^{-(n-c)/2},
\end{equation}
which  is clearly bounded as $n$ grows.

If instead $(f_1,f_2,f_3,f_4)=(i,k,j,l)$ or $(f_1,f_2,f_3,f_4)=(i,k,l,j)$,
 one gets
\begin{equation}
\frac1{n^2}  \sum_{1\le a<c\le n}\sum_{1\le b<c\le n}\frac{2^{n+a+b+c-3}}{B_c^{n+a+b+c-3}B_c^{n+c-1} B_c^{n+c-1}},
\end{equation}
which is easily seen to be $O(1/n^2)$.

All the other permutations of $(i,j,k,l)$ give a contribution which
equals, by symmetry, one of the three we just considered.

\item the tree $\mathcal T^{(n)}_{\{i,j,k,l\}}$ has two branches: one of
  them does not bifurcate, the other one bifurcates into two sub-branches, one of which bifurcates into two
sub-sub-branches, {\sl cf.} Figure
  \ref{fig:2t}(b). Let $a_1,a_2,a_3$ be the levels where the three bifurcations occur, ordered so that $1\le a_1<a_2<a_3\le n$. This time,
we choose $f_1$ between $1$ and $2^n$ and then, for $i=1,2,3$, $f_{i+1}$ among the $2^{a_i-1}$ leaves which join with $f_1$ at level $a_i$.
If $(f_1,f_2,f_3,f_4)=(i,j,k,l)$ one has in this case
\begin{multline}
\phantom{move}
\frac1{n^2}  \sum_{1\le a_1<a_2<a_3\le n}\frac{2^{n+a_1+a_2+a_3-3}}{B_c^{n+a_1+a_2+a_3-3}B_c^{n+a_1-1} B_c^{n+a_3-1}}\,
=
\\
\frac1{\sqrt 2 n^2} \sum_{1\le a_1<a_2<a_3\le n}2^{-(n-a_2)/2},
\end{multline}
which is $O(1/n)$. 
Finally, when $(f_1,f_2,f_3,f_4)$ is equal to $(i,k,j,l)$ or 
to $(i,k,l,j)$ one gets
  \begin{multline}
  \phantom{move}
\frac1{n^2}  \sum_{1\le a_1<a_2<a_3\le n}\frac{2^{n+a_1+a_2+a_3-3}}{B_c^{n+a_1+a_2+a_3-3}B_c^{n+a_2-1} B_c^{n+a_3-1}}\,=
\\
\frac1{\sqrt 2 n^2} \sum_{1\le a_1<a_2<a_3\le n}2^{-(n-a_1)/2},
\end{multline}
which is $O(1/n^2)$.
\end{enumerate}
\end{proof}

\section{Marginal relevance of disorder: the non-hierarchical case}\label{sec:nHproof}

Here we prove Theorem \ref{th:main2} and therefore we assume that
\eqref{eq:K} holds with $\alpha=1/2$.

We choose and fix once and for all a $\gamma \in (2/3,1)$ and set for $h>0$
\begin{equation}
\label{eq:kh1}
  k:=k(h):=\left\lfloor \frac 1h\right\rfloor.
\end{equation}

\smallskip

\begin{remark} \rm
\label{rem:k}
In \cite{cf:DGLT}
 the choice $k(h)= \lfloor 1/\tf(0,h)\rfloor $ was made and 
 it corresponds to choosing $k(h)$ equal to the correlation length
 of the annealed system. In our case $1/\tf(0,h) \stackrel{h\searrow 0}
 \sim 1/(c_K h^2)$ ({\sl cf.} \eqref{eq:annF}) and therefore  \eqref{eq:kh1}
 may look surprising. However, there is nothing particularly deep behind: 
 for  $\ga=1/2$, due to the fact that we have to prove delocalization
 for $h\le  \exp(-c_7/\gb^4)$, choosing $k(h)$ that diverges for small
 $h$ like $1/h$ instead of $1/h^2$ just leads to choosing $c_7$
 different by a factor $2$ (and we do not track the precise value of constants). 
 We take this occasion to stress that it is practical to work always
 with sufficiently large values of $k(h)$, and this can be achieved
 by choosing $c_7$ sufficiently large.
\end{remark}

\smallskip

We divide $\N$ into blocks
\begin{equation}
  \label{eq:blocks}
  B_i:=\{(i-1)k+1,(i-1)k+2,\ldots,ik\}\;\mbox{with\;\;}i=1,2,\ldots .
\end{equation}
From now on we assume that $(N/k)$ is integer, and of course it is also the
number of blocks contained in the interval $\{1,\ldots,N\}$.

We define, in analogy with the hierarchical case, 
\begin{equation}
  A_N\, :=\, \bbE\left(Z_{N,\go}^\gamma\right),
\end{equation}
and we note that, as in \eqref{eq:JensAn}, Jensen's inequality implies that a
sufficient condition for $\tf(\beta,h)=0$ is that $A_N$ does not
diverge when $N\to\infty$.  Therefore, our task is to show that 
for every $\gb_0>0$ we can find $c_7>0$ such that
for every $\beta\le \beta_0$ and $h$ such that
\begin{equation}
  \label{eq:assh}
0<h\le \exp(-c_7/\beta^4), 
\end{equation}
one has that
$\sup_N A_N<\infty$.

\subsection{Decomposition of $Z_{N,\go}$ and change of measure}
The first step is a decomposition of the partition function similar to
that used in \cite{cf:T_cg}, which is a refinement of the strategy employed 
in \cite{cf:DGLT}.  For $0<i\le j$ we let
$Z_{i,j}:=Z_{(j-i),\theta^i\go}$, with $(\theta^i\go)_a:=\go_{i+a},
a\in \N$, {\sl i.e.}, $\theta^i\go$ is the result of the application to
$\go$ of a shift by $i$ units to the left.  We decompose $Z_{N,\go}$
according to the value of the first point ($n_1$) of $\tau$ after $0$,
the last point ($j_1$) of $\tau$ not exceeding $n_1+k-1$, then the
first point ($n_2$) of $\tau$ after $j_1$, and so on. We call $i_r$
the index of the block in which $n_r$ falls, 
and $\ell:=\max\{r:n_r\le N\}$,
see Figure \ref{fig:decomp}. Due to the constraint $N\in\tau$, one has always $i_\ell=(N/k)$.
\begin{figure}[ht]
\begin{center}
\leavevmode
\epsfxsize =12.4 cm
\psfragscanon
\psfrag{0}[c]{\small $0$}
\psfrag{ui}[c]{\tiny $i_0=j_0=0$}
\psfrag{N}[c]{\small N} 
\psfrag{B1}[c]{\small $B_1$}
\psfrag{B2}[c]{\small $B_3$}
\psfrag{B3}[c]{\small $B_4$}
\psfrag{B8}[c]{\small $B_9$}
\psfrag{B9}[c]{\small $B_{10}$}
\psfrag{B10}[c]{\small $B_{11}$}
\psfrag{B13}[c]{\small $B_{14}$}
\psfrag{Z1}[c]{\tiny $Z_{n_1,j_1}$}
\psfrag{Z2}[c]{\tiny $Z_{n_2,j_2}$}
\psfrag{Z3}[c]{\tiny $Z_{n_3,j_3}$}
\psfrag{Z4}[c]{\tiny $Z_{n_4,N}$}
\psfrag{n1}[c]{ \small $n_1$}
\psfrag{j1}[c]{ \small $j_1$}
\psfrag{n2}[c]{ \small $n_2$}
\psfrag{j2}[c]{ \small $j_2$}
\psfrag{n3}[c]{ \small $n_3$}
\psfrag{j3}[c]{ \small $j_3$}
\psfrag{n4}[c]{ \small $n_4$}
\psfrag{j1}[c]{ \small $j_1$}
\psfrag{k}[c]{\small $k$}
\psfrag{2k}[c]{\small $2k$}
\psfrag{n1+k}[c]{ \tiny $n_1+k$}
\psfrag{n2+k}[c]{ \tiny $n_2+k$}
\psfrag{n3+k}[c]{ \tiny $n_3+k$}
\epsfbox{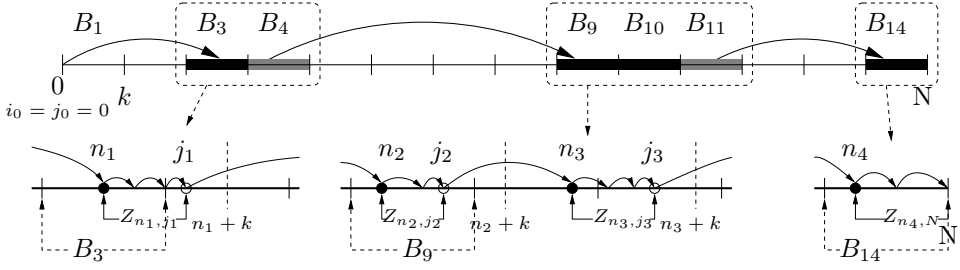}
\end{center}
\caption{\label{fig:decompos} A typical configuration which
  contributes to $\hat Z_\go^{(i_1,\ldots,i_\ell)}$.  In this example
  we have $N/k=14$, $\ell=4$, $i_1=3$, $i_2=9$, $i_3=10$ and
  $i_4=N/k=14$ (by definition $i_\ell =N/k$, {\sl cf.}
  \eqref{eq:dec1}).  Contact points are only in black and grey blocks:
  the blocks $B_{i_j}$, $j=1,\ldots, \ell$ are black and they contain
  one (and only one) point $n_i$. To the
  right of a black block there is either another black block or a grey
  block (except for the last black block, $B_{i_\ell}$, that contains
  the end-point $N$ of the system).  The bottom part of the figure
  zooms on black and grey blocks.  We see that 
  to the right of $n_{i}$ (big black dots) there are
  renewal points before $n_{i}+k$; for $i<\ell$,  $j_i$ is the rightmost one and
  it is marked by a big empty dot (even if it is not the case in the
  figure, it may happen that there is none: in that case $j_i=n_i$).
  Therefore, between empty dots and black dots there is no contact
  point (the origin should be considered an empty dot too).  Note that
  $j_i$ can be in $B_i$, as it is the case for $j_2$, or in $B_{i+1}$,
  as it is the case for $j_1$ and $j_3$.  Going back to the figure on
  top, we observe that the set $M$ of \eqref{eq:M} is
  $\{3,4,9,10,11,14\}$, that is the collection of black and grey
  blocks. We point out that it may happen that a grey block contains
  no point, but it is convenient for us to treat grey blocks as if
  they always contained contact points. It is only to the charges
  $\go$ in black and grey blocks that we apply the change-of-measure
  argument that is crucial for our proof. }
\label{fig:decomp}
\end{figure}

In formulas:
\begin{equation}
\label{eq:dec1}
  Z_{N,\go}=\sum_{\ell=1}^{N/k}\sum_{i_0:=0<i_1<\ldots<i_\ell=N/k}
\hat Z_\go^{(i_1,\ldots,i_\ell)},
\end{equation}
where
\begin{multline}
  \label{eq:Zhat}
  \hat Z_\go^{(i_1,\ldots,i_\ell)}\, :=\, 
  \sum_{n_1\in B_{i_1}}\sum_{j_1=n_1}^{n_1+k-1}\sumtwo{n_2\in B_{i_2}:}{n_2\ge
    n_1+k}\sum_{j_2=n_2}^{n_2+k-1}\ldots \sumtwo{n_{\ell-1}\in B_{i_{\ell-1}}:}
{n_{\ell-1}\ge
    n_{\ell-2}+k}\sum_{j_{\ell-1}=n_{\ell-1}}^{n_{\ell-1}+k-1}\sumtwo
{n_\ell\in B_{N/k}:}{n_\ell\ge n_{\ell-1}+k}\\
 z_{n_1}K(n_1)Z_{n_1,j_1}z_{n_2}K(n_2-j_1)
Z_{n_2,j_2}\, \ldots \, z_{n_\ell}K(n_\ell-j_{\ell-1})Z_{n_\ell,N},
\end{multline}
and $z_n:=e^{\beta\go_n+h-\beta^2/2}$.

Then, from inequality \eqref{eq:fracineq}, we have
\begin{equation}
\label{eq:A_nh}
  A_N\,\le\, \sum_{\ell=1}^{N/k}\sum_{i_0:=0<i_1<\ldots<i_\ell=N/k}
\bbE\left[(\hat Z_\go^{(i_1,\ldots,i_\ell)})^\gamma\right],
\end{equation}
and, 
as in \eqref{eq:hold}, we apply H\"older's inequality to get
\begin{multline}
  \label{eq:Hold}
  \bbE \left[\left(\hat Z^{(i_1,\ldots,i_\ell)}_\go\right)^\gamma\right]\,
=\\
\tilde \bbE\left[\left(\hat Z^{(i_1,\ldots,i_\ell)}_\go\right)^\gamma
\frac{\dd \bbP}{\dd \tilde \bbP}(\go)\right]
\le \left(
\tilde \bbE \hat Z^{(i_1,\ldots,i_\ell)}_\go\right)^\gamma
\left(\bbE \left[\left(\frac{\dd \bbP}{\dd \tilde \bbP}\right)^{\gamma/(1-\gamma)}\right]
\right)^{1-\gamma}.
\end{multline}
The new law $\tilde\bbP:= \tilde\bbP^{(i_1,\ldots,i_\ell)}$ will be
taken to depend on the set $(i_1,\ldots,i_\ell)$. In order to define
it, let first of all
\begin{equation}
  \label{eq:M}
  M:=M(i_1,\ldots,i_\ell):=\{i_1,i_2,\ldots,i_\ell\}\cup \{i_1+1,i_2+1,\ldots,
i_{\ell-1}+1\}.
\end{equation}
Then, we say that under $\tilde\bbP $ the random vector $\go$  is Gaussian, centered 
and with covariance matrix
\begin{equation}
  \label{eq:covarianze}
  \tilde \bbE(\go_i\go_j)\, =\, \ind_{i=j}-\mathcal C_{ij}\,:=\,
  \begin{cases}
   \ind_{i=j}-H_{ij} & \text{ if there exists } u\in M \text{ such that } i,j \in B_u,
   \\
    \ind_{i=j} & \text{ otherwise,}
  \end{cases}
\end{equation}
and 
\begin{equation}
  H_{ij}\,:=\, 
    \begin{cases}
(1-\gamma)/\sqrt{9\,  k(\log k) \;|i-j|}&\text{ if } i\ne j,
\\
0& \text{ if } i=j.      
\label{eq:H}
    \end{cases}
\end{equation}
Note that all the $\mathcal C_{ij}$'s are non-negative.
It is immediate to check that the $k\times k$ symmetric matrix
$\hat H:=\{H_{ij}\}_{i,j=1}^k$ satisfies
\begin{equation}
  \label{eq:HS_H}
\Vert\hat H\Vert\, :=\,
\sqrt{\sum_{i,j=1}^k H_{ij}^2}\, \le \frac{1-\gamma}2, 
\end{equation}
for $k$ sufficiently large. 
In words: $\omega_n$'s in different
blocks are independent; in blocks $B_u$ with $u\notin M$ they are just
IID standard Gaussian random variables, while if $u\in M$
then the random vector $\{\go_n\}_{n\in B_u}$ has
covariance matrix $I-\hat H$, where $I$ is the $k\times k$ identity matrix.
Note that, since $\Vert\hat H\Vert$ dominates the spectral
radius of $\hat H$, \eqref{eq:HS_H} guarantees 
  that $I -\hat H$ is positive definite (and also that 
  $I -(1-\gamma)^{-1}\hat H$ is positive definite, that will be needed just below).

The last factor in the right-hand side of \eqref{eq:Hold} is easily
obtained recalling \eqref{eq:dets} and independence of the $\go_n$'s
in different blocks, and one gets
\begin{equation}
\label{eq:RN_nh}
\left(\tilde \bbE\left[ \left(\frac{\dd \bbP}{\dd \tilde \bbP}\right)^{\gamma/(1-\gamma)}\right]
\right)^{1-\gamma}=\left(\frac{\det(I-\hat H)}
{(\det(I-1/(1-\gamma)\hat H))^{1-\gamma}}
\right)^{|M|/2}.
\end{equation}
Since $\hat H$ has trace zero and its (Hilbert-Schmidt) norm satisfies
\eqref{eq:HS_H}, one can apply $\det(I-\hat H)\le
\exp(-\mbox{Trace}(\hat H))=1$ and \eqref{eq:dano} (with $V$ replaced by
$\hat H$ and $\gep$ by $1$) to get that the right-hand side of
\eqref{eq:RN_nh} is bounded above by $ \exp(|M|/2)$, which in turn is 
bounded by $\exp(\ell)$.  Together with \eqref{eq:Hold} and
\eqref{eq:A_nh}, we conclude that
\begin{equation}
  A_N\le\sum_{\ell=1}^{N/k}\sum_{i_0:=0<i_1<\ldots<i_\ell=N/k}
e^{{\ell}} \left[\tilde \bbE \hat Z_\go^{(i_1,\ldots,i_\ell)}
\right]^\gamma.
\label{eq:sqbr}
\end{equation}

\subsection{Reduction to a non-disordered model  }

We wish to  bound  the right-hand side of \eqref{eq:sqbr} 
with the
partition function of a non-disordered pinning model in the
delocalized phase, which goes to zero for large $N$.  We start by
claiming that 
\begin{multline}
\label{eq:claimU}
  \tilde \bbE \hat Z_\go^{(i_1,\ldots,i_\ell)}\, \le\,  \sum_{n_1\in B_{i_1}}\ldots \sumtwo{n_\ell\in B_{N/k}:}{n_\ell\ge 
n_{\ell-1}+k} K(n_1)K(n_2-j_1)\ldots K(n_\ell-j_{\ell-1})\\
\times U(j_1-n_1)U(j_2-n_2)\ldots U(N-n_\ell),
\end{multline}
where
\begin{equation}
\label{eq:U}
  U(n)\, =\, c_{8}\bP(n\in\tau)\bE\left[e^{-\beta^2\sum_{1\le i<j\le {n/2}}H_{ij}\delta_i
\delta_j}
\right],
\end{equation}
and $c_8$ is a positive constant depending only on $K(\cdot)$. This is
proven in Appendix \ref{sec:appprova}. We are also going to make use
of:

\medskip

\begin{lemma}
\label{th:condlemma}
\label{th:lemma_cg}
   There exists $C_2=C_2(K(\cdot))<\infty$ such that if, for some $\eta>0$,
  \begin{equation}
    \label{eq:condlemma1}
    \sum_{j=0}^{k-1}U(j)\le \eta  {\sqrt k}
  \end{equation}
and 
\begin{equation}
  \label{eq:condlemma2}
  \sum_{j=0}^{k-1}\sum_{n\ge k}U(j)K(n-j)\le \eta,
\end{equation}
then there exists $C_1=C_1(\eta, k,K(\cdot))$ 
such that
the right-hand side of \eqref{eq:claimU} is  bounded above by 
\begin{equation}
  C_1 \eta^\ell C_2^{\ell}\prod_{r=1}^\ell\frac1{(i_r-i_{r-1})^{3/2}}.
\end{equation}
\end{lemma}
\medskip
It is important to note that $C_2$ does not depend on $\eta$.

Lemma~\ref{th:condlemma} is a small variation on 
 \cite[Lemma 3.1]{cf:T_cg}, but,
 both because the model we are considering is somewhat different
 and for sake of completeness, we give
 the details
of the proof in Appendix \ref{sec:appprova}.
\smallskip

Now assume that conditions \eqref{eq:condlemma1}-\eqref{eq:condlemma2}
are verified for some $\eta$. Collecting \eqref{eq:sqbr},
\eqref{eq:claimU} and Lemma \ref{th:lemma_cg}, we have then
\begin{equation}
  A_N\le C_1^\gamma\sum_{\ell=1}^{N/k}\sum_{i_0:=0<i_1<\ldots<i_\ell=N/k}
\left(\eta^\gamma\,C_2^{\gamma}e 
\right)^\ell\prod_{r=1}^\ell\frac1{(i_r-i_{r-1})^{(3/2)\gamma}}.
\label{eq:puremod}
\end{equation}
In the right-hand side we recognize, apart from the irrelevant
multiplicative constant $C_1^\gamma$, the partition function of a
non-random $(\beta=0)$ pinning model with $N$ replaced by $N/k$,
$K(\cdot)$ replaced by
\begin{equation}
  \hat K(n)\, =\, \frac1{n^{(3/2)\gamma}}\frac1{\sum_{i\ge1}i^{-(3/2)\gamma}},
\end{equation}
and $h$ replaced by
\begin{equation}
  \label{eq:hbar}
  \hat h:=\log \left(\eta^\gamma\,C_2^{\gamma}e
  \sum_{n\in\N}\frac1{n^{(3/2)\gamma}}\right).
\end{equation}
Note that $\hat K(\cdot)$ is normalized to be a probability measure on $\N$,
which is possible since (by assumption) $\gamma>2/3$, and that it has a
power-law tail with exponent $(3/2)\gamma>1$. Thanks to Lemma
\ref{th:puro} below, one has that the right-hand side of \eqref{eq:puremod} tends to zero
for $N\to\infty$ whenever
\begin{equation}
\label{eq:condeta}
\hat h<0.
\end{equation}
Therefore, if  $\eta$ is so small that \eqref{eq:condeta} holds,
we can conclude that $A_N$ tends to zero for $N\to\infty$ and therefore $\tf(\beta,h)=0$.

The proof of Theorem \ref{th:main2} is therefore concluded once we prove

\medskip
\begin{proposition}
\label{th:stimeU}
Fix $\eta>0$ such that \eqref{eq:condeta} holds. For every $\gb_0>0$
there exists $0<c_7<\infty$ such that if $\beta\le \gb_0$ and $0<h\le
\exp(-c_7/\beta^4)$, conditions
\eqref{eq:condlemma1}-\eqref{eq:condlemma2} are verified.
\end{proposition}
\medskip

\noindent
\begin{proof}[Proof of Proposition \ref{th:stimeU}]
We need to show that the two hypotheses of Lemma~\ref{th:condlemma} hold
and for this 
we are going to use the following result:
\medskip

\begin{lemma}
\label{th:CE}
Under the law $\bP$, the  random variable 
\begin{equation}
W_L\, :=\,  (\sqrt{L} \log L)^{-1}
\sum_{1\le i<j\le L} \delta_i
        \delta_j /\sqrt{j-i},
\end{equation}        
        converges in distribution, as $L$ tends to $\infty$, to $c \vert Z \vert$
        ($Z\sim N(0,1)$ and $c$ a positive constant).
\end{lemma}
\medskip

This lemma, 
the proof of which may be found just below (together with the explicit value of $c$), directly implies that,
if we set $S(a,L):=  \bE\left[\exp\left(- a W_L\right)\right]$, we have
 $\lim_{a\to\infty} \lim_{L\to\infty}S(a,L)\, = \, 0$ and, by the monotonicity
 of $S (\cdot, L)$, we get
\begin{equation}
  \label{eq:small}
  \lim_{a, L\to\infty}   S(a,L)\, = \, 0.
\end{equation}

Let us  verify \eqref{eq:condlemma1}.
Note first of all that ({\sl cf.} \eqref{eq:U} and \eqref{eq:H})
\begin{multline}
  U(n)\, =\, c_8 \bP(n\in\tau)
     S\left(\beta^2(1-\gamma)\sqrt{\log k}\,\sqrt{\frac {n/2}{9\,  k}}\,
        \frac{\log (n/2)}{
          {\log k}}
        , \frac n2\right)\\
         =:\, c_8 \bP(n\in\tau) s_\gb(k,n).
\end{multline}
We recall also that
(\cite[Th. B]{cf:Doney})
\begin{equation}
\label{eq:doney}
  \bP(n\in\tau)\stackrel{n\to\infty}\sim \frac1{2\pi C_K \sqrt n},
\end{equation}
and therefore there exists $c_9>0$ such that for every $n \in \N$
\begin{equation}
\label{eq:doney-bound}
  \bP(n\in\tau)\, \le \,  \frac{c_9}{\sqrt {n}}.
\end{equation}
Split the sum in \eqref{eq:condlemma1} according to whether
$j \le \gd k$ or not ($\gd=\gd(\eta) \in (0,1)$
is going to be chosen below).
By using $S(a,L)\le 1$ (in the case $j \le \gd k$) and  \eqref{eq:doney-bound}
we obtain
\begin{equation}
\label{eq:interm4}
\sum_{j=0}^{k-1} U(j) \, \le \, c_8+c_8 c_9 \sum_{j=1}^{\gd k}
\frac1{\sqrt{j}} \, +\, c_8 c_9 \sum_{j= \gd k +1}^{k-1}
\frac1{\sqrt{j}}\, s_\gb(k,j).
\end{equation}
Since if $c_7$ is chosen sufficiently large
\begin{equation}
  \beta^2\sqrt{\log k}\ge\sqrt{c_7-\beta^4\log 2}\ge \sqrt{c_7}/2, 
\end{equation}
and since $k$ may be made large by increasing $c_7$, 
we directly see that \eqref{eq:small} implies
that $s_\gb(k,j)$ may be made  smaller than (say) 
$\gd$ for every $\delta k <j<k$ by choosing $c_7$ sufficiently large. Therefore \eqref{eq:interm4} implies
 \begin{equation}
\label{eq:interm5}
\sum_{j=0}^{k-1} U(j) \, \le \, 4c_8 c_9 (\sqrt{\gd} + \gd) \sqrt{k}.
\end{equation}       
By choosing $\gd=\gd(\eta) $ such that         $4c_8 c_9 (\sqrt{\gd} + \gd) \le \eta$,
we have \eqref{eq:condlemma1}.
The proof of \eqref{eq:condlemma2} is absolutely analogous to the proof
of \eqref{eq:condlemma1} and it is therefore omitted. 
\end{proof}

\subsection{Proof of Lemma~\ref{th:CE}}
We introduce the notation
\begin{equation}
Y_L^{(i)}\, :=\, \sum_{j=i+1}^{L} \frac{\gd_j}{\sqrt{j-i}}
,\text{\  so that \  }
W_L \, =\, \frac 1{\sqrt{L}\log L} \sum_{i=1}^{L-1}\gd_i Y_L^{(i)}.
\end{equation}
Let us observe that, thanks to the renewal property of $\tau$, under
$\bP (\cdot \vert \gd_i=1)$, $Y_L^{(i)}$ is distributed like
$Y_{L-i}:= Y_{L-i}^{(0)}$ (under $\bP$).  
 The first step in the proof is  observing that, in view of \eqref{eq:doney-bound},
\begin{multline}
  \bE \left[\frac 1{\sqrt{L}\log L} \sum_{i=(1-\gep)L }^{L-1}\gd_i Y_L^{(i)}\right]\,
=\\
\frac1{\log L \sqrt L}\sum_{i=(1-\gep)L}^{L-1}\sum_{j=i+1}^L\frac{
\bP(i\in\tau)\bP(j-i\in\tau)}{\sqrt{j-i}}
\, 
=
\, O(\gep),
\end{multline}
uniformly in $L$,
so we can focus on studying $W_{L, \gep}$, defined as $W_L$, but stopping the 
sum over $i$ at $(1-\gep) L$. At this point we use that
\begin{equation}
\label{eq:Chung-Erdos}
\lim_{L \to \infty} \frac{Y_L}{\log L} \, =\, \frac 1{2\pi C_K} \, =: \hat c_K,
\end{equation}
in $L^2(\bP)$ (and hence in $L^1(\bP)$). We postpone the proof of \eqref{eq:Chung-Erdos}
and observe that, thanks to the properties of the logarithm,
it implies that for every $\gep>0$
\begin{equation}
\label{eq:use}
\lim_{L \to \infty}
\sup_{q \in [\gep, 1]} \bE \left[ \left \vert 
\frac 1{\log L} \sum_{j=1}^{qL} \frac {\gd _j}{\sqrt{j}} \, - \hat c_K \right \vert \right]\, =\, 0.
\end{equation}

Let us write
\begin{equation}
R_L \, := W_{L, \gep} \, -\, \frac{\hat c_K} {\sqrt{L} } \sum_{i=1}^{(1-\gep) L} \gd _i
\end{equation}
and note that $ L^{-1/2} \sum_{i=1}^{(1-\gep) L} \gd _i$ converges in law toward
 $\sqrt{(1-\gep)/(2\pi C_K^2)} \, \vert Z\vert$.
 This follows 
  directly by using that the event $ \sum_{i=1}^{L} \gd _i \ge  m$
 is the event $\tau_m \le L$ ($\tau_m$ is of course the 
$m$-th point in $\tau$ after $0$) and by using the fact that $\tau_1$ is in
 the domain of attraction of the positive stable law of index $1/2$
 \cite[VI.2 and XI.5]{cf:Feller2}.
It suffices therefore to show that $\bE [\vert R_L\vert] $ tends to zero. 
We have
\begin{multline}
  \bE \left[\vert R_L\vert\right]\, \le\, 
  \frac {1}{\sqrt{L}}   \sum_{i=1}^{(1-\gep) L}
  \bE [ \gd_i ] \bE\left [ \left.\bigg\vert \frac{Y^{(i)}_L}{\log L}-\hat c_K \bigg \vert \;\right\vert
  \delta_i=1 \right] \, =\, 
  \\
  \frac {1}{\sqrt{L}}   \sum_{i=1}^{(1-\gep) L}
  \bE [ \gd_i ] \bE \left[ \left\vert \frac{Y_{L-i}}{\log L}-\hat c_K \right\vert 
  \right] \, =\, o(1), 
\end{multline}
where in the last step we have used \eqref{eq:use} and \eqref{eq:doney-bound}.

Note that we have also proven that $c=(2\pi)^{-3/2} C_K^{-2}$ in the statement of Lemma \ref{th:CE}.

\medskip

We are therefore left with the task of proving \eqref{eq:Chung-Erdos}.
This result has been already proven \cite[Th.~6]{cf:chungerdos} when
$\tau$ is given by the successive returns to zero of a centered,
aperiodic and irreducible random walk on $\bbZ$ with bounded variance
of the increment variable. Note that, by well established local limit
theorems, for such a class of random walks we have \eqref{eq:doney}.
Actually in \cite{cf:chungerdos} it is proven that
\eqref{eq:Chung-Erdos} holds almost surely as a consequence of
$\text{var}_\bP(Y_L)=O(\log L)$.  What we are going to do is
simply to re-obtain such a bound, by repeating the steps in
\cite{cf:chungerdos} and using \eqref{eq:doney}-\eqref{eq:doney-bound}, for the general
renewal processes that we consider (as a side remark: also in our
generalized set-up, almost sure convergence holds).

The proof goes as follows:
by using \eqref{eq:doney} it is straightforward to see that the limit
as $L\to \infty$ of 
$\bE [Y_L/\log L]$ is $\hat c_K$, so that we are done
if we show that $\text{var}_\bP(Y_L/\log L)$ vanishes as $L\to \infty$.
So we start by observing that
\begin{equation}
\label{eq:CEst1}
\text{var}_\bP(Y_L) \, =\, \sum_{i, j} \frac{\bE[ \gd_i \gd_j]-
\bE[ \gd_i] \bE[ \gd_j]
 }{\sqrt{ij}} \, =\, 2 \sum_{i=1}^{L-1}\sum_{j=i+1}^L
  \frac{\bE[ \gd_i \gd_j]-
\bE[ \gd_i] \bE[ \gd_j]
 }{\sqrt{ij}}
 + O(1),
\end{equation}
by \eqref{eq:doney-bound}. Now we compute 
\begin{equation}
\begin{split}
\sum_{i=1}^{L-1}\sum_{j=i+1}^L
  \frac{\bE[ \gd_i \gd_j]-
\bE[ \gd_i] \bE[ \gd_j]
 }{\sqrt{ij}} 
 \, &=\,
 \sum_{i=1}^{L-1} \frac{\bE[\gd_i]}{\sqrt{i}}
 \left[\sum_{j=1}^{L-i} \frac{\bE[\gd_j]}{\sqrt{j+i}}-
 \sum_{j=i+1}^{L} \frac{\bE[\gd_j]}{\sqrt{j}}
 \right]
 \\
 &\le \, 
 \sum_{i=1}^{L-1} \frac{\bE[\gd_i]}{\sqrt{i}}
 \left[\sum_{j=1}^{L-i} \frac{\bE[\gd_j]}{\sqrt{j+i}}-
 \sum_{j=i+1}^{L} \frac{\bE[\gd_j]}{\sqrt{j+i}}
 \right]
 \\
 &\le \, 
 \sum_{i=1}^{L-1} \frac{\bE[\gd_i]}{\sqrt{i}}
 \sum_{j=1}^{i} \frac{\bE[\gd_j]}{\sqrt{j+i}}
 \\
 \le \, \sum_{i=1}^{L-1} &\frac{\bE[\gd_i]}{i} 
  \sum_{j=1}^{i} {\bE[\gd_j]} \le  c_9^2  
  \sum_{i=1}^{L-1} \frac{1}{i^{3/2}} 
  \sum_{j=1}^{i} \frac 1{j^{1/2}} = O(\log L),
 \end{split}
\end{equation}
where, in the last line, we have used \eqref{eq:doney-bound}.
In view of \eqref{eq:CEst1}, we have obtained 
$\text{var}_\bP(Y_L)=O(\log L)$ and
the proof  \eqref{eq:Chung-Erdos}, and therefore of Lemma~\ref{th:CE} is complete.
\qed

        

\appendix
\section{Some technical results and useful estimates}

\subsection{Two results on renewal processes}
The first result concerns the non-disordered pinning model and is well known:
\begin{lemma}
\label{th:puro}
  Let $ K(\cdot)$ be a probability on $\N$ which satisfies \eqref{eq:K} for
some $\alpha>0$. If $ h<0$, we have
that 
\begin{equation}
\label{eq:lemma1}
\sum_{\ell=1}^N \sum_{i_0:=0<i_1<\ldots<i_\ell=N}
e^{h\ell}\prod_{r=1}^\ell K(i_r-i_{r-1})\stackrel {N\to\infty}\longrightarrow0.  
\end{equation}
\end{lemma}
\medskip

This is implied by \cite[Th. 2.2]{cf:Book}, since the left-hand side of 
\eqref{eq:lemma1} is nothing but the partition function of the homogeneous
pinning model of length $N$, whose critical point is $h_c=0$ ({\sl cf.} also 
\eqref{eq:annF}).

\bigskip

The second fact we need is
\medskip

\begin{lemma}
\label{th:condiz}
There exists a positive constant $c$, which depends only on
$K(\cdot)$, such that for every positive function $f_N(\tau)$ which depends
only on $\tau\cap\{1,\ldots,N\}$ one has
  \begin{equation}
  \label{eq:condiz}
\sup_{N>0}    \frac{\bE[f_N(\tau)|2N\in\tau]}{\bE[f_N(\tau)]}\le c.
  \end{equation}
\end{lemma}

\medskip

\noindent
\begin{proof} 
The statement follows by writing $f_N(\tau)$
as $f_N(\tau)\sum_{n=0}^N \ind_{\{X_N=n\}}$, where $X_N$ is the last renewal epoch up
to (and including) $N$, and using the bound
$$
\sup_N \max_{n=0,\ldots,N}\frac{
\bP(X_N=n \vert 2N \in \tau)}{\bP(X_N=n )} =:c<\infty ,$$
which  is equation (A.15) in \cite{cf:DGLT} 
(this has been proven also in  \cite{cf:T_cg}, where the proof is repeated to show 
that $c$ can be chosen as a function of $\ga$ only).
\end{proof}

\smallskip

\subsection{Proof of \eqref{eq:claimU}}
\label{sec:appprova}
Defining the event
\begin{equation}
  \label{eq:Omega}
  \Omega_{\underline n,\underline j}:=\{N\in\tau\;\;\mbox{and}\;\;
\{j_{r-1},\ldots,n_r\}\cap \tau=\{j_{r-1},n_r\}\;\;\mbox{for all}\;\;
  r=1,\ldots,\ell\}, 
\end{equation}
with the convention that $j_0:=0$, 
we have
\begin{equation}
  \hat Z_\go^{(i_1,\ldots,i_\ell)}=\sum_{n_1\in B_{i_1}}\ldots \sumtwo{n_\ell\in B_{N/k}:}{n_\ell\ge 
n_{\ell-1}+k}\bE\left[e^{\sum_{n=1}^N(\beta\go_n+h-\beta^2/2)\delta_n};\Omega_{\underline n,
\underline j}\right].
\end{equation}
Since $\tilde\bbP$ is a Gaussian measure and $\delta_i^2=\delta_i$ for 
every $i$,
the computation of $\tilde\bbE \hat Z_\go^{(i_1,\ldots,i_\ell)}$ is
immediate:
\begin{equation}
  \tilde\bbE \hat Z_\go^{(i_1,\ldots,i_\ell)}=\sum_{n_1\in B_{i_1}}\ldots \sumtwo{n_\ell\in B_{N/k}:}{n_\ell\ge 
n_{\ell-1}+k}\bE\left[
e^{h\sum_{n=1}^N\delta_n-\beta^2/2\sum_{i,j=1}^N \mathcal C_{ij}\delta_i\delta_j}
    ;\Omega_{\underline n,
      \underline j}\right].
\end{equation}
In view of $\mathcal C_{ij}\ge 0$, we obtain an upper bound by
neglecting in the exponent the terms such that $n_r\le i\le j_{r}$ and
$n_{r'}\le j\le j_{r'}$ with $r\ne r'$. At that point, the $\bE$
average may be factorized, by using the renewal property,
 and we obtain (recall that $\mathcal C_{ii}=0$)
\begin{equation}
\begin{split}
  \tilde\bbE\hat Z_\go^{(i_1,\ldots,i_\ell)}\, \le\, & \sum_{n_1\in B_{i_1}}\ldots \sumtwo{n_\ell\in B_{N/k}:}{n_\ell\ge 
    n_{\ell-1}+k}K(n_1)\ldots K(n_\ell-j_{\ell-1})\\
  &\times \prod_{r=1}^\ell\bE\left[\left.e^{h\sum_{i=n_r}^{j_r}\delta_i-\beta^2\sum_{n_r\le i< j\le j_r}\mathcal C_{ij}
        \delta_i\delta_j}\ind_{\{j_{r}\in\tau\}}\right|n_r\in \tau\right], 
\end{split}
\end{equation}
with the convention that $j_\ell:=N$.
We are left with the task of proving that
\begin{equation}
\label{eq:auxU}
  \bE\left[\left.e^{h\sum_{i=n_r}^{j_r}\delta_i-\beta^2\sum_{n_r\le i<j\le j_r}\mathcal C_{ij}
\delta_i\delta_j}\ind_{\{j_{r}\in\tau\}}\right|n_r\in \tau\right]\le U(j_r-n_r),
\end{equation}
with $U(\cdot)$ satisfying \eqref{eq:U}.
We remark first of all that the left-hand side of \eqref{eq:auxU} equals
\begin{equation}
  \bP(j_r-n_r\in\tau)\bE\left[\left.e^{h\sum_{i=n_r}^{j_r}\delta_i-\beta^2\sum_{n_r\le i<j\le j_r}\mathcal C_{ij}
\delta_i\delta_j}\right|n_r\in \tau,j_r\in\tau\right].
\end{equation}
Since by construction $j_r-n_r<k(h)=\lfloor 1/h\rfloor$, 
one has
\begin{equation}
e^{h\sum_{i=n_r}^{j_r}\delta_i}\, \le\,  e.
\end{equation}
As for the remaining average, assume without loss
of generality that $|\{n_r,n_r+1,\ldots,j_r\}\cap B_{i_r}|\ge (j_r-n_r)/2$
(if this is not the case, the inequality clearly holds with $B_{i_r}$
replaced by $B_{i_r+1}$ and the  arguments which follow are trivially
modified).  Then,
\begin{multline}
\label{eq:uffa}
  \bE\left[\left.e^{-\beta^2\sum_{n_r\le i< j\le j_r}\mathcal C_{ij}
\delta_i\delta_j
}\right|n_r\in \tau,j_r\in\tau
\right]\, \le
\\
 \bE\left[\left.
\exp\left(-\beta^2\sum_{0<i<j\le (j_r-n_r)/2}\delta_i\delta_j H_{ij}\right)
\right| j_r-n_r\in\tau \right].
\end{multline}
Finally, 
 the conditioning in \eqref{eq:uffa} can be eliminated using Lemma \ref{th:condiz}, and
\eqref{eq:claimU} is proved.
\qed

\subsection{Proof of Lemma \ref{th:condlemma} }

In this proof (and in the statement) two positive numbers $C_1$ and $C_2$ appear. 
$C_1$ is going to change along with the steps
of the proof: it depends on $\eta$, $k$ and on $K(\cdot)$. 
$C_2$ instead is chosen once and for all below and it 
depends only on $K(\cdot)$. 
We start by giving a name to  the right-hand side of \eqref{eq:claimU}:
\begin{multline}
  \label{eq:restart}
  Q\, :=\, 
  \sum_{n_1\in B_{i_1}}\sum_{j_1=n_1}^{n_1+k-1}\sumtwo{n_2\in B_{i_2}:}{n_2\ge
    n_1+k}\sum_{j_2=n_2}^{n_2+k-1}\ldots \sumtwo{n_{\ell-1}\in B_{i_{\ell-1}}:}
{n_{\ell-1}\ge
    n_{\ell-2}+k}\sum_{j_{\ell-1}=n_{\ell-1}}^{n_{\ell-1}+k-1}\sumtwo
{n_\ell\in B_{N/k}:}{n_\ell\ge n_{\ell-1}+k}\\
 K(n_1)\ldots K(n_\ell-j_{\ell-1})
U(j_1-n_1)\ldots   U(j_{\ell-1}-n_{\ell-1}) U(N-n_\ell)
 .
\end{multline}
Since $N -n _\ell <k$, we can get rid of $U(N-n_\ell)\, (\le c_8
 \bP(N-n_\ell \in \tau) )$ and of the right-most sum (on $n_\ell$),
 replacing $n_\ell$ by $N$, by paying a price that depends on $k$ and
 $K(\cdot)$ (this price goes into $C_1$).  Therefore we have
 \begin{equation}
  \label{eq:lems1}
 Q\, \le \, 
 C_1\,  \sum_{n_1\in B_{i_1}} \ldots
 \sum_{j_{\ell-1}=n_{\ell-1}}^{n_{\ell-1}+k-1}
 K(n_1)\ldots K(n_\ell-j_{\ell-1})
U(j_1-n_1)\ldots U(j_{\ell-1}-n_{\ell-1})
 ,
\end{equation}
where by convention from now on $n_\ell:=N$.
Now we single out the long jumps. The set of long jump arrival points is
defined as
\begin{equation}
J\, =\, J( i_1,i_2, \ldots, i_\ell)\, :=\, \left\{r:\, 1 \le r \le \ell, \, i_r>i_{r-1}+2\right\},
\end{equation}
and the definition guarantees that a long jump $\{j_{r-1},\ldots,n_r\}$ contains at least one whole
block with no renewal point inside. For $r\in J$ 
we use the bound
\begin{equation}
K(n_r-j_{r-1}) \, \le \, \frac{C_2}{(i_r-i_{r-1})^{3/2} k^{3/2}}, 
\end{equation}
and we stress that we may and do choose $C_2$ depending only on $K(\cdot)$.
For later use, we choose $C_2\ge 2^{3/2}$.
This leads to
\begin{multline}
  \label{eq:lems3}
  Q\, \le \, 
 C_1\,  {k^{-3|J| /2}}
 \prod_{r \in J} \frac{C_2}{(i_r-i_{r-1})^{3/2}}
 \\
\times \sum_{n_1\in B_{i_1}} \ldots
 \sum_{j_{\ell-1}=n_{\ell-1}}^{n_{\ell-1}+k-1}
 \left( 
 \prod_{r\in \{1,\ldots, \ell\}\setminus J} 
 K(n_r- j_{r-1})
 \right)
 \,
U(j_1-n_1)\ldots U(j_{\ell -1}-n_{\ell-1}).
\end{multline}
Now we perform the sums in \eqref{eq:lems3} and bound the outcome by
using the assumptions \eqref{eq:condlemma1} and \eqref{eq:condlemma2}.

We first sum over $j_{r-1}$, $r \in J$, keeping of course
into account the constraint $0\le j_{r-1}-n_{r-1}<k$. By using
 \eqref{eq:condlemma1} such sum yields at most $(\eta \sqrt{k})^{\vert J\vert}$
 if $1 \notin J$. If $1 \in J$, for $r=1$ then $j_0=0$ and there is no summation: 
 we can still bound the sum by $(\eta \sqrt{k})^{\vert J\vert}$,
 provided that we change the constant $C_1$.

Second, we sum over $j_{r-1},n_r$ for 
$r \in  \{1,\ldots, \ell\}\setminus J$ and use 
\eqref{eq:condlemma2}. Once again we have to treat separately the case
$r=1$, as above. But if $1 \notin \{1,\ldots, \ell\}\setminus J$ we directly see that
the summation is bounded by $\eta^{\ell -\vert J \vert}$.

Finally, we have to sum over $n_r$, for $r\in J$. The summand does not depend
on these variables anymore, so this gives at most $k^{\vert J \vert }$.

Putting these estimates together we obtain
\begin{equation}
\label{eq:lems4}
Q\, \le \, 
 C_1\,  
 \frac{
 (\eta \sqrt{k})^{\vert J\vert }
 \eta^{\ell -\vert J \vert  } k^{\vert J \vert}  
   }
 {k^{3|J| /2}}
 \prod_{r \in J} \frac{C_2}{(i_r-i_{r-1})^{3/2}}\, \le\, 
 C_1 \eta^\ell C_2^\ell 
 \prod_{r=1}^{\ell} \frac{1}{(i_r-i_{r-1})^{3/2}},
\end{equation}
where, in the last step, we have used $C_2\ge 2^{3/2}$. 
The proof of Lemma~\ref{th:condlemma} is therefore 
complete.
\qed



\section*{Acknowledgments}
We are very grateful to Bernard Derrida for many enlightening discussions and to an anonymous referee for having observed
the link between hierarchical pinning and Galton-Watson processes.
The authors acknowledge the support of ANR, grant POLINTBIO. F.T. was partially
supported also by ANR, grant LHMSHE.


\frenchspacing
\bibliographystyle{plain}

\end{document}